\documentclass[format=acmsmall, review=false]{acmart}

\usepackage{acme20} 

\usepackage[utf8]{inputenc}
\usepackage{amsmath}
\usepackage{array}
\usepackage{graphicx}
\usepackage[english]{babel}
\usepackage[utf8]{inputenc}
\usepackage{algorithm}
\usepackage[noend]{algpseudocode}

\usepackage{enumitem}
\usepackage{wasysym}
\usepackage{framed}
\usepackage{pgfgantt,rotating}
\usepackage{subfloat}
\usepackage{multirow}
\usepackage{blindtext}

\title{The Economic Limits of Permissionless Consensus}

\author{Eric Budish}
\affiliation{%
  \institution{\department{Booth School of Business} \institution{University of Chicago} \city{Chicago} \country{USA}}
}
\email{eric.budish@chicagobooth.edu}

\author{Andrew Lewis-Pye}
\affiliation{%
  \institution{\department{Department of Mathematics} \institution{London School of Economics} \city{London} \country{UK}}
}
\email{a.lewis7@lse.ac.uk}

\author{Tim Roughgarden}
\affiliation{%
  \institution{\department{Computer Science Department}
    \institution{Columbia University \& a16z Crypto} \city{New York} \country{USA}}
}
\email{tim.roughgarden@gmail.com}

\date{Feb 2024}

\begin{abstract} 
The purpose of a consensus protocol is to keep a distributed network
of nodes “in sync,” even in the presence of an unpredictable
communication network and adversarial behavior by some of the
participating nodes. In the permissionless setting relevant to modern
blockchain protocols, these nodes may be operated by a large number of
unknown players, with each player free to use multiple identifiers and
to start or stop running the protocol at any time. Establishing that a
permissionless consensus protocol is "secure" thus requires both a
distributed computing argument (that the protocol guarantees consistency
and liveness unless the fraction of adversarial participation is
sufficiently large) and an economic argument 
(that 
carrying out an attack would be
prohibitively expensive for a potential attacker). There is a mature
toolbox for assembling arguments of the former type; the goal of this
paper is to lay the foundations for arguments of the latter type. For
example, the Ethereum protocol is oft-claimed to be "more economically
secure" after "the merge," meaning in its current proof-of-stake
incarnation relative to the (proof-of-work) original. What, formally,
does this assertion mean? Is it true? Could there be alternative
protocols that are "still more economically secure" than Ethereum? How
do the answers depend on the assumptions imposed on, for example, the
reliability of message delivery or the active participation of
non-malicious players?

An ideal permissionless consensus protocol would, in addition to satisfying standard consistency and liveness guarantees, render consistency violations prohibitively expensive for the attacker without collateral damage to honest participants---for example, by programatically confiscating an attacker's resources without reducing the value of honest participants' resources, as is the intention for slashing in a proof-of-stake protocol. We make this idea precise with our notion of the EAAC (expensive to attack in the absence of collapse) property,
and prove the following results: 
\begin{itemize}
\item In the synchronous and dynamically available setting (in which
the communication network is reliable but non-malicious players may be
periodically inactive), with an adversary that controls at least
one-half of the overall resources, no protocol can be EAAC. 
In particular, this result rules out EAAC for all typical longest-chain protocols (be they proof-of-work or proof-of-stake).

\item In the partially synchronous and quasi-permissionless setting
(in which resource-controlling non-malicious players are always active
but the communication network may suffer periods of unreliability),
with an adversary that controls at least one-third of the overall
resources, no protocol can be EAAC. In particular, slashing in a proof-of-stake protocol cannot achieve its intended purpose if message delays cannot be bounded a priori.

\item In the synchronous and quasi-permissionless setting, there is a
proof-of-stake protocol with slashing that, provided the adversary controls less
than two-thirds of the overall stake, satisfies the EAAC property.

\end{itemize}
Thus, while only ``classical security'' is possible in the dynamically
available or partially synchronous settings, proof-of-stake protocols
with slashing can obtain additional ``economic security'' 
in the quasi-permissionless and synchronous setting.
All three results are optimal with respect to the size of the
adversary. 

With respect to Ethereum, our work formalizes the potential security
benefits of proof-of-stake sybil-resistance coupled with slashing and
the common belief that the merge has increased Ethereum's economic
security. Our work also provides mathematical justifications for
several key design decisions behind the post-merge Ethereum protocol,
ranging from long cooldown periods for unstaking to economic penalties
for inactivity.

\end{abstract}

\begin{document}

\begin{titlepage}

\maketitle

\end{titlepage}

\section{Introduction}\label{s:intro}

\subsection{The Security of Permissionless Consensus Protocols}\label{ss:intro}

The core functionality of a blockchain protocol such as Bitcoin or
Ethereum is {\em permissionless consensus}, with a potentially large
and ever-evolving set of participants kept in sync on the state of the
blockchain via a consensus protocol.
Compared to the traditional setting of permissioned consensus
protocols (with a fixed and known participant set), permissionless
protocols must cope with three novel challenges (cf., \cite{lewis2023permissionless}):
\begin{itemize}

\item \textbf{The unknown players challenge.}
The set of participants is unknown at the time of protocol deployment
and is of unknown size.

\item \textbf{The player inactivity challenge.} Participants can start
  or stop running the protocol at any time.

\item
\textbf{The sybil challenge.} One participant may masquerade as many
by using many identifiers (a.k.a.\ ``sybils'').

\end{itemize}

Reasoning about the security of a \emph{permissioned} consensus protocol is,
to a large extent, a purely computer science question. 
(Here ``security'' means
that the protocol satisfies both liveness and consistency---as long as
there's work to be done it gets done, and without any two participants
ever committing to conflicting decisions.)  Consider, for example, a
corporation that wants to grant customers access to a database while
keeping its downtime percentage to 0.0001\%. One approach would be
to replicate the database on many different servers and use a
consensus protocol to keep those replicas in sync. How many replicas
are necessary and sufficient to achieve the desired uptime? Or, in two
parts:
\begin{itemize}

\item [(Q1)] For a given value of~$k$, how many replicas~$n(k)$ are
  necessary and sufficient to guarantee security even when~$k$ of the
  replicas have failed?

\item [(Q2)] For a given target downtime percentage~$\delta$, what
  is the smallest value of~$k$ such that the probability that at
  least~$k+1$ of~$n(k)$ replicas fail simultaneously is at
  most~$\delta$?

\end{itemize}
The distributing computing literature resolves question~(Q1) for a
staggering variety of settings (with~$n(k)=2k+1$ and $n(k)=3k+1$ being
two of the most common answers); see, e.g., \cite{cachin2011introduction,lynch1996distributed} for an
introduction. Given a probabilistic model of replica failures,
question~(Q2) then boils down to a calculation, providing the
corporation with the appropriate value of~$k$ and the corresponding
number~$n(k)$ of servers that it should buy.

Reasoning about the security of a \emph{permissionless} consensus protocol
fundamentally requires the synthesis of computer science and economic
arguments.
First, the sybil challenge generally forces such a protocol to measure
``size'' in terms of some resource that is, unlike identifiers, scarce
and therefore costly. Common examples include hashrate (as in a
proof-of-work protocol) and staked cryptocurrency (as in a
proof-of-stake protocol).
Second, while misbehaving replicas in a permissioned protocol are
generally attributed to hardware failures and software bugs,
permissionless protocols must tolerate deliberately malicious behavior
by a motivated attacker (a hacker, designers of a competing protocol, or even an
unfriendly nation-state).
The analogs of questions~(Q1) and~(Q2) are then:
\begin{itemize}

\item [(Q3)] What is the largest value of~$\rho$ for which a protocol
  can guarantee security even when a $\rho$ fraction of the costly
  resource is controlled by an attacker?

\item [(Q4)] How unlikely is it that an attacker controls more
  than a $\rho$ fraction of the costly resource and then carries
  out an attack?

\end{itemize}
Question~(Q3) is well defined, and the last decade of research on
blockchain protocols has answered it in a range of settings.
Question~(Q4) makes no sense without an economic model for the cost of carrying out an attack. The goal of this paper is to develop a mathematical framework for quantifying this cost and for designing permissionless consensus protocols in which this cost is as large as possible.

\subsection{The Economic Consequences of an Attack: Scorched Earth vs.\ Targeted Punishment}\label{ss:scorched}

In the Bitcoin white paper, Nakamoto~\cite{nakamoto2008bitcoin} noted that an
attacker controlling 51\% of the overall hashrate could force
consistency violations and thereby carry out double-spend attacks, but
suggested that the consequent economic cost might make such an attack
unprofitable:
\begin{quote}
  If a greedy attacker is able to assemble more CPU power than all the
  honest nodes, he would have to choose between using it to defraud
  people by stealing back his payments, or using it to generate new
  coins. He ought to find it more profitable to play by the rules,
  such rules that favour him with more new coins than everyone else
  combined, than to undermine the system and the validity of his own
  wealth.
\end{quote}
This argument rests on the assumption that a double-spend attack
would cause a significant and permanent drop in the USD-denominated
market price of Bitcoin's native currency~BTC, with the attacker then
foregoing most of the USD-denominated value of the future BTC block
rewards that it's positioned to receive.

More recently, this initial narrative around the security of Bitcoin and other
proof-of-work blockchain protocols has evolved into a second narrative, for two
reasons. First, empirical evidence for the ``double-spends will crash
the cryptocurrency price'' hypothesis has been weak. Second, the ``CPUs'' that Nakamoto referred to have been almost entirely replaced
by ASICs that serve no purpose other than to evaluate a hard-coded
cryptographic hash function such as SHA-256.  Now, the story goes: if
an attacker with 51\% of the hashrate were to carry out a double-spend
attack on (say) the Bitcoin protocol, the Bitcoin ecosystem could
respond with the ``nuclear option,'' changing the cryptographic hash
function used for proof-of-work mining via a coordinated upgrade to
the protocol (a ``hard fork''). Such an upgrade would render existing
ASICs useless, leaving the attacker with a defunct pile of scrap
metal. Hopefully, the mere threat of this nuclear option would deter any
potential attackers, and the option would never have to actually be
exercised.

An alarming aspect of both narratives
is the ``scorched earth'' nature of an attack's consequences: honest
participants (passive holders of BTC or ASIC-owning honest miners,
respectively) are harmed as much as the attacker. 
%
%
Is scorched earth-style punishment fundamental to permissionless
consensus, or an artifact of the specific design decisions made in
Bitcoin and other proof-of-work protocols?

Ideally, a blockchain protocol could punish an attacker that carries
out a double-spend attack in a targeted and non-scorched-earth way,
leaving honest participants unharmed. The hope for such ``asymmetric
punishment'' has long rested with proof-of-stake blockchain protocols,
in which the ``power'' of a participant is proportional to how much
of the protocol's native currency they have locked up in a
designated staking contract.  Intuitively, with the scarce and costly
resource controlled directly by the protocol (rather than
``off-chain,'' as with hashrate), such a protocol is positioned to
directly and surgically confiscate resources from specific
participants (perhaps as part of a hard fork, or perhaps programmatically
as part of the protocol's normal operation)~\cite{slasher,kwon}. 

The Ethereum protocol, which famously migrated from proof-of-work to
proof-of-stake (among many other changes) in September~2022 in an event known as ``the merge,''
offers an interesting case study. Ethereum's lead founder, Vitalik Buterin, wrote in the early design stages that 
``The `one-sentence philosophy' of proof of stake is \ldots `security comes from putting up economic value-at-loss'{''}~\cite{buterin2016blog} and ``The intention is to make 51\% attacks extremely expensive, so that even a majority of validators working together cannot roll back finalized blocks without undertaking an extremely large economic loss''~\cite{buterin2017blog}.  
Today, post-merge, the protocol's official
documents echo the aspirations above, asserting that ``proof-of-stake offers greater crypto-economic
security than proof-of-work'' and ``economic penalties for misbehavior
make 51\% style attacks more costly for an attacker compared to
proof-of-work.''\footnote{See
  \url{https://ethereum.org/developers/docs/consensus-mechanisms/pos}.}
What would be a rigorous phrasing of these assertions? Are they true?
Could there be alternative protocols---perhaps wildly different from
the proof-of-work or proof-of-stake protocols considered to
date---that are ``still more economically secure'' than Ethereum? How
do the answers depend on the assumptions imposed on, for example, the
reliability of message delivery or the active participation of
non-malicious players?


\subsection{Overview of Results}\label{ss:overview}

\paragraph{Defining the cost of an attack.}
We augment the permissionless consensus framework of Lewis-Pye and Roughgarden~\cite{lewis2023permissionless} with a new model for reasoning about the cost of causing consistency violations.
Informally, we call such an attack {\em cheap} if the attacker suffers no economic consequences following the consistency violation, and {\em expensive} otherwise. We call the attack {\em expensive due to collapse} if the attacker is harmed for the wrong (``scorched earth'') reasons, with honest protocol participants also suffering, as in the two narratives around the security of proof-of-work protocols described in Section~\ref{ss:scorched}.
We call the attack {\em expensive in the absence of collapse} if the attacker is harmed for the right reasons---targeted punishment that avoids collateral damage to the honest participants, as is the goal of slashing in a proof-of-stake protocol.
Our key definition is that of an {\em EAAC} protocol (for ``expensive to attack in the absence of collapse'') which states, informally, that every attack on the protocol is expensive in the absence of collapse.
We then interpret a protocol that guarantees the EAAC property 
as ``more economically secure'' than one that does not.
We next survey the three main results in this paper (depicted in Figure~\ref{f:lattice})---two that
identify settings in which only ``classical security'' is possible,
and a third showing that, under appropriate assumptions,
proof-of-stake protocols with slashing can indeed obtain additional
``economic security.''

\begin{figure}[h]
\centering
\includegraphics[width=.95\textwidth]{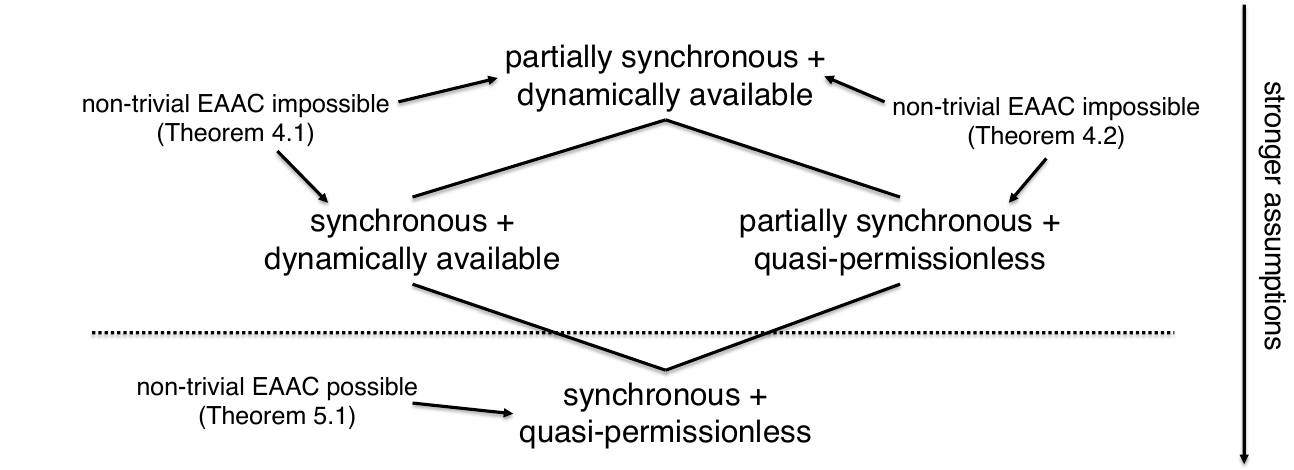}
\caption{Summary of main results. While only ``classical security'' is
  achievable in the dynamically available setting (in which
  non-malicious players may be periodically inactive) or the partially
  synchronous setting (in which the communication network may suffer
  unbounded periods of unreliability), proof-of-stake protocols with
  slashing can achieve additional ``economic security'' in the
  quasi-permissionless and synchronous setting.}\label{f:lattice}
\end{figure}

\paragraph{Impossibility of EAAC in the dynamically available setting.}
Our first main result (Theorem~\ref{neg}) states that, in the
synchronous and dynamically available setting (in which the
communication network is reliable but non-malicious players may be
periodically inactive), with an adversary that controls at least
one-half of the overall resources, no protocol can be EAAC.  This
result is optimal with respect to the size of the adversary, as the
Bitcoin protocol (among others) guarantees (probabilistic) consistency
and liveness in the synchronous and dynamically available setting when
the adversary controls less than one-half of the overall
hashrate~\cite{garay2018bitcoin}, and is thus vacuously EAAC in this
case (Figure~\ref{f:eaac1}).  In particular, this result rules out
non-trivial EAAC guarantees for all typical longest-chain protocols
(be they proof-of-work protocols like Bitcoin or proof-of-stake
protocols such as Ouroboros \cite{kiayias2017ouroboros} and Snow White
\cite{daian2019snow}).

\begin{figure}[h]
\centering
\includegraphics[width=.9\textwidth]{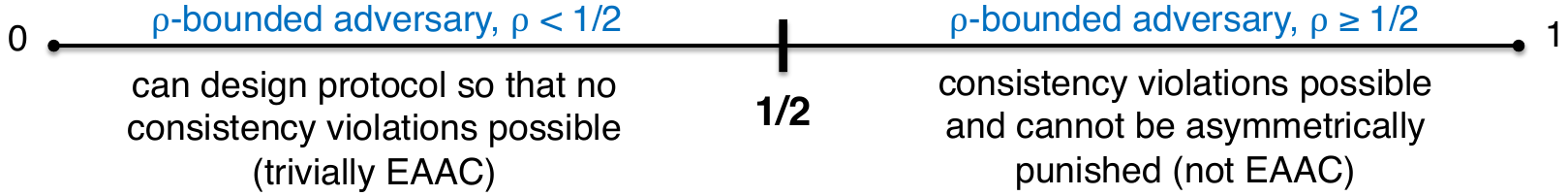}
\caption{Theorem~\ref{neg}. No non-trivial EAAC guarantees are possible in the
  dynamically available setting, even with synchronous
  communication: once an adversary is large enough to cause
  consistency violations, it is also large enough to avoid asymmetric
  punishment. A $\rho$-bounded adversary is one that
controls at most a $\rho$ fraction of each resource (such as hashrate
or stake)  used by a
protocol.}\label{f:eaac1}
\end{figure}

To give the flavor of the proof, consider two disjoint sets of players
$X$ and $Y$ that each own an equal amount of resources. Liveness in
the dynamically available setting implies that if one of these sets
never hears from the other, it must forge ahead and continue to
confirm transactions. So imagine that the players in~$X$ are
malicious, don't talk to~$Y$, and behave as if they were honest and
never heard from~$Y$, confirming transactions to themselves that
conflict with the transactions confirmed by~$Y$ during the same
period. Now suppose that, at some later point, players in $X$
disseminate all the messages that they would have disseminated if
honest in their simulated execution. At this point, it is not possible
for late-arriving players to determine whether the players in $X$ or
the players in $Y$ are honest. If the protocol happens to make this
particular attack expensive (by harming the players in~$X$), there is
a symmetric execution (with the players in~$X$ honest and those in~$Y$
malicious) in which the honest players are the ones who are harmed.

\paragraph{Impossibility of EAAC in the partially synchronous setting.}
Our second result (Theorem~\ref{neg2}) states that, in the partially synchronous and quasi-permissionless setting
(in which resource-controlling non-malicious players are always active
but the communication network may suffer periods of unreliability),
with an adversary that controls at least one-third of the overall
resources, no protocol can be EAAC.\footnote{This version of the theorem statement assumes that the protocol is required to be live when the adversary is $\rho$-bounded for $\rho<1/3$. The full theorem statement in Section \ref{res1} generalises the result to consider weaker liveness requirements.} 
(In fact, something stronger is true: 
no protocol can make consistency violations expensive for the
attacker, even allowing for collateral damage to honest players. 
In particular, slashing in a proof-of-stake protocol cannot achieve its intended purpose if message delays cannot be bounded a priori.
This result is optimal with respect to the size of the
adversary, as there is a proof-of-stake PBFT-style
protocol that guarantees consistency and liveness in the partially
synchronous and quasi-permissionless setting when the adversary
controls less than one-third of the overall stake~\cite{lewis2023permissionless}, which is therefore vacuously
EAAC in that regime (Figure~\ref{f:eaac2}).

\begin{figure}[h]
\centering
\includegraphics[width=.9\textwidth]{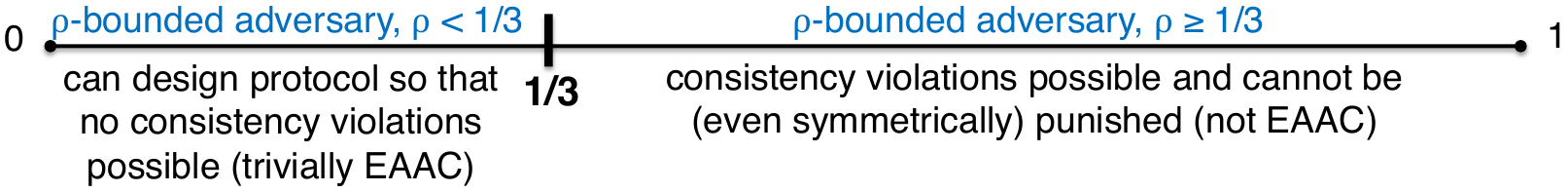}
\caption{Theorem~\ref{neg2}. No non-trivial EAAC guarantees are possible with
 partially synchronous communication, even in the quasi-permissionless
 setting: once an adversary is large enough to cause consistency
 violations, it is also large enough to avoid (even symmetric) punishment.}\label{f:eaac2}
\end{figure}

To give a sense of the proof of this result, consider three sets of
players~$X$, $Y$, and~$Z$, all with equal resources, with the players
in~$X$ and~$Z$ honest and the players in~$Y$ malicious. Suppose that
messages disseminated by players in $X$ are received by players in
$X \cup Y$ right away but not by players in~$Z$ for a very long time
(which is a possibility in the partially synchronous
model). Symmetrically, suppose $Y$ and $Z$ but not $X$ promptly
receive messages sent by players in~$Z$. Suppose further that the
malicious players in~$Y$ pretend to the players in~$X$ that they've
never heard from anyone in~$Z$ and to the players in~$Z$ that they've
never heard from anyone in~$X$. Liveness dictates that the players
in~$X$ and the players in~$Z$ must each forge ahead and confirm
transactions, even though no messages between players in~$X$ and~$Z$
have been delivered yet. These uncoordinated confirmed transactions
will generally conflict, resulting in a consistency violation.
Moreover, this violation may not be noticed by the players of~$X$
and~$Z$ for a very long time (again due to the arbitrarily long delays
in the partially synchronous model), giving the players of~$Y$ the
opportunity to sell off their resources and avoid any possible
punishment in the meantime.


\paragraph{Possibility of EAAC in the synchronous and quasi-permissionless setting.}
Our final result (Theorem~\ref{posres}) states that, in the
synchronous and quasi-permissionless setting, there is a
proof-of-stake protocol with slashing that, provided the adversary
controls less than two-thirds of the overall stake, satisfies the EAAC
property.\footnote{Later in this section (and in more detail in Section \ref{posproof}), we will discuss how the 2/3 bound can be improved by reducing liveness requirements. }  In fact, our protocol is designed for a version of the
synchronous setting defined by two known parameters: one
parameter~$\Delta$ that represents typical network speed, perhaps on
the order of milliseconds or seconds; and a second
parameter~$\Delta^*$ that represents the time required for players to
communicate when the network is unreliable (whether over the network
or by out-of-band coordination), perhaps on the order of days or
weeks.\footnote{One interpretation of the parameter~$\Delta^*$ is as
  the speed of communication through social channels, as referenced
  in, for example,
  Buterin's discussion of``weak subjectivity''~\cite{vitalikblog}.}
Intuitively, when there's no attack, the protocol operates at a speed
proportional to~$\Delta$; under attack, it recovers at a speed
proportional to~$\Delta^*$.  

More precisely, we give a protocol that:
(i) is live and consistent in the partially synchronous setting (with
respect to the parameter~$\Delta$) provided the adversary controls
less than one-third of the overall stake; and (ii) is EAAC so long as
the adversary controls less than two-thirds of the overall stake and
message delays are no larger than $\Delta^*$.  A recent result of Tas
et al.~\cite{tas2023bitcoin}, when translated to our model, implies
that our result is optimal with respect to the size of the adversary:
even in the permissioned and synchronous setting, if a protocol
guarantees liveness with respect to an adversary that controls less
than one third of the overall resources, it cannot guarantee EAAC
with respect to an adversary that controls at least two-thirds of the
overall resources (Figure~\ref{f:eaac3}).
%
%

\begin{figure}[h]
\centering
\includegraphics[width=.9\textwidth]{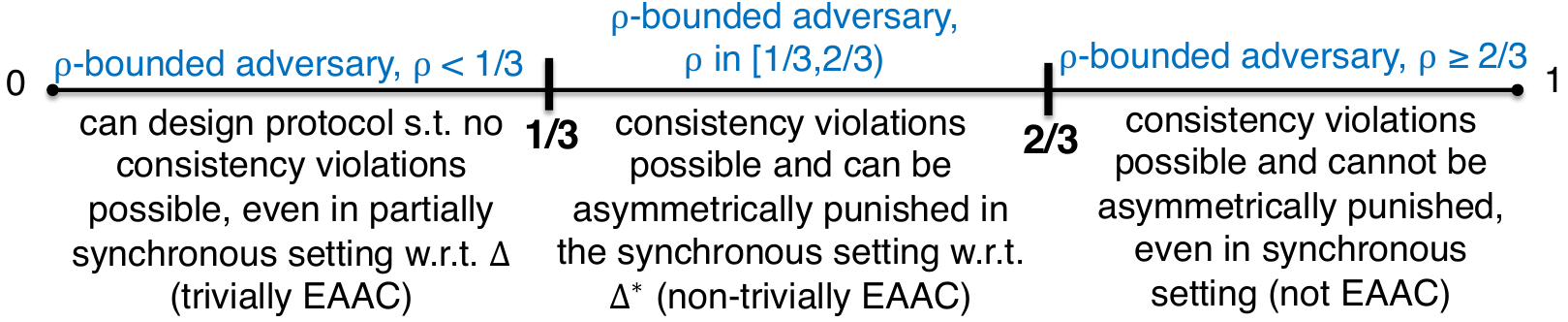}
\caption{Theorem~\ref{posres}. Non-trivial EAAC guarantees are possible in the
  quasi-permissionless and synchronous setting. The parameter $\Delta$
  is an upper bound on
message delays during ``normal operation,'' while~$\Delta^*$ bounds
the time required for honest players to communicate (over the
network or out-of-band) when the protocol is under attack.}\label{f:eaac3}
\end{figure}

Achieving asymmetric punishment via slashing in a proof-of-stake
protocol would seem to require addressing the following challenges.\label{page:challenges}
\begin{enumerate}

\item There should be a ``smoking gun'' behind every consistency
  violation, in the form of a ``certificate of guilt'' that identifies (at
  least some of) the Byzantine players responsible for the
  violation. 

\item All honest players should learn such a certificate of guilt promptly
  after a consistency violation, before the adversary has had the
  opportunity to cash out its assets.

\item Given the prompt receipt of a certificate of guilt, honest players
 should be able to reach consensus on a new (post-slashing) state.

\end{enumerate}
Further, the second and third challenges must be met despite
interference from an adversary that is so large as to be able to cause
consistency violations.

At a very high level, our protocol resolves these three challenges as
follows.  (Section~\ref{six} provides a more technical overview of the
protocol, with the full details deferred to Section~\ref{posproof}.
The protocol is designed to be as simple as possible subject to a
non-trivial EAAC guarantee; optimizing the performance of such
protocols is an interesting direction for future work.)
Some PBFT-style protocols 
provide certificates of guilt in the form of
votes on conflicting blocks, and accordingly our starting point is the
Tendermint protocol~\cite{buchman2016tendermint,BKM18}. (Extending the
permissioned Tendermint protocol to a proof-of-stake protocol with
guaranteed consistency and liveness in the quasi-permissionless
setting is non-trivial, but we show that it can be done.) The standard
Tendermint protocol uses two stages of voting per view, and we show
that this fails to guarantee the prompt receipt of certificates of guilt by
honest players.  On the other hand, we show that, provided the
adversary controls less than two-thirds of the stake, three stages of
voting suffice for the prompt dissemination of certificates of
guilt. Finally, after a consistency violation, honest players 
attempt to reach consensus on an updated genesis block (in which slashing has
been carried out) using a variant of the Dolev-Strong
protocol~\cite{dolev1983authenticated}. 

\paragraph{The trade-off between liveness and the EAAC property.} Thus far, our discussion of Theorem \ref{posres} assumes that we require protocols to be live and consistent in the partially synchronous setting  provided the adversary controls
less than one-third of the overall stake. A standard technique (which just modifies the definition of a `quorum certificate' to require votes from a larger proportion of relevant players)  can be used to extend the regime in which our protocol is EAAC beyond the 2/3 bound at the expense of reducing liveness resilience: If the protocol  is required to be live for adversaries that are $\rho_l$-bounded (for some $\rho_l<1/3$), then it can be made EAAC when the adversary is $\rho$-bounded for  $\rho<1-\rho_l$.  See Section \ref{posproof} for further details.

\paragraph{Interpretation for Ethereum.} Our results make precise a
number of common beliefs about the Ethereum protocol (post-merge).
Most
obviously, our work formalizes the potential security benefits of
proof-of-stake sybil-resistance coupled with slashing logic.  

Our impossibility result for the dynamically available setting
(Theorem~\ref{neg}) shows that such security guarantees are impossible
both for any protocol that relies on only off-chain resources (such as
proof-of-work protocols) and for any standard longest-chain
protocol (even if it is a proof-of-stake protocol). Thus, two of the
biggest changes made to the Ethereum protocol during the merge---the
switch from proof-of-work to proof-of-stake, and the addition of the
Casper finality gadget~\cite{buterin2017casper}---are both necessary for provable slashing
guarantees (neither change alone would enable asymmetric punishment).
Our result also provides a justification for the ``inactivity leaks''
used in the Ethereum protocol to punish seemingly inactive players,
which can be viewed as an economic mechanism for enforcing the
(necessary) assumptions of the quasi-permissionless setting.

Our impossibility result for the partially synchronous setting
(Theorem~\ref{neg2}) justifies the common assumption that honest
Ethereum validators could, in the case of emergency, communicate out
of band within some known finite amount of time. Further, this result
justifies the long ``cooldown period'' for unstaking in post-merge
Ethereum, with a delay that is roughly proportional to the assumed
time required for out-of-bound communication.

The protocol we design to prove our positive result (Theorem~\ref{posres})
resembles post-merge Ethereum in several high-level respects: the use
of proof-of-stake sybil-resistance, an accountable PBFT-type approach
to consensus (in our case, Tendermint rather than Casper), slashing
for asymmetric punishment, equal-size validators, and regularly
scheduled updates to the validator set. One notable difference is our
protocol's reliance on three stages of voting to ensure the prompt
dissemination of certificates of guilt, for the reasons discussed in
Section~\ref{furtherc}.  


\subsection{Discussion}\label{ss:flow}

\paragraph{The economic cost of controlling resources.}
This paper follows the tradition of Nakamoto~\cite{nakamoto2008bitcoin} (as in the quote in Section~\ref{ss:scorched}), among many others, in concentrating on the economic consequences to an attacker of causing a consistency violation (e.g., in order to execute a double-spend attack) in a permissionless consensus protocol. Our focus complements previous work of Budish~\cite{budish2018economic,budish2023trust}, which provides an economic model for quantifying the cost that a potential attacker must absorb to control sufficient resources to be able to cause a consistency violation.
The main conclusion in
Budish~\cite{budish2018economic,budish2023trust} is that, for the
current major blockchain protocols, this cost is surprisingly cheap,
scaling as a flow cost rather than a stock cost. This analysis
suggests that the economic consequences that \emph{follow} an attack
are the ones with the greatest potential to deter an attacker, and in
this sense provides strong economic justification for the focus of
this paper.


\paragraph{Budish's analysis.}
Some readers may find it instructive to reinterpret ``expensive'' and
``cheap'' attacks in the ``stock vs.\ flow'' language of Budish~\cite{budish2018economic,budish2023trust};
other readers can skip this paragraph and the next without
loss of continuity.
Budish's analysis was originally 
framed in terms of proof-of-work protocols but applies more generally
to protocols that use arbitrary external or on-chain resources,
including proof-of-stake protocols. 
For example, consider a protocol
that uses a single type of resource (e.g., ASICs in Bitcoin or coins
in Ethereum), and let $c$ denote the (flow) cost per timeslot to supply one
unit of resources.  For simplicity, assume that all protocol
participants face the same cost. For a proof-of-work protocol, $c$ is
the cost per timeslot to run one ASIC, including variable costs such
as electricity and an ongoing cost of capital applied to the
hardware. So, in this case, $c$ may be specified by $c = rC + \eta$,
where $C$ is the (stock) cost of one unit of hardware, $rC$ is the capital
cost of one unit of hardware per timeslot (with $r$ the per-timeslot
interest rate), and $\eta$ is the variable cost of electricity per
unit of hardware per timeslot. For a proof-of-stake protocol, $c$
represents the opportunity cost of locking up stake, so $c$ may be
specified by $c = rC$ where $C$ is the cost of one unit of stake.

Let $d$ denote the total (expected) reward earned by all validators
combined per timeslot (perhaps a combination of inflationary rewards
paid out by the protocol, transaction fees paid by users, and
additional ``MEV'' extracted from the application layer). For a
blockchain protocol like Bitcoin or Ethereum, the barriers to enter or
exit the set of miners or validators is low; thus, we expect the
marginal miner or validator to break even. Mathematically, this
translates to the zero-profit condition that $c = d/N^*$ or
equivalently $N^*c=d$, where~$N^*$ denotes the number of resource
units held by protocol participants at equilibrium. Because the
reward~$d$ is paid out to miners or validators every single timeslot,
it is presumably a small fraction of the total value of the protocol's
native currency. Now suppose that causing a consistency violation
requires the attacker to hold~$AN^*$ units of the resource, say for
some $A < 2$. If an attack takes $T$ timeslots to complete, then the
cost incurred by the attacker in holding sufficient capital for the
attack is only $AN^{\ast}cT=AdT<2dT$.  (This bound holds even after
setting aside the rewards that the attacker would earn during these
$T$ timeslots and whatever off-chain value it might obtain from its
double-spend attack.) If $d$ and~$T$ are small and there are no
economic consequences to the attacker subsequent to its attack, the
conclusion is that consistency violations and double-spend attacks are
relatively cheap, with cost scaling with the recurring per-block
costs~$d=N^*c$ of securing the blockchain protocol, rather than the
total value~$N^*C$ of the resources devoted to securing the protocol.
The EAAC property introduced in the present paper loosely corresponds
to the insistence that the economic cost of causing a consistency
violation should scale with~$C$ rather than with~$c$.

\paragraph{Consistency vs.\ liveness violations.} One reason we focus
on consistency violations (as opposed to liveness violations) is that
they are the ones that can be unequivocally attributed to deviations
from the intended protocol by an attacker (as opposed to network
delays or other vagaries outside the control of honest
participants). A second is that, in practice, consistency violations
are generally considered much more serious than liveness violations
(akin to money permanently disappearing from your bank account vs.\
the bank's computer system going down for a few hours).

\subsection{Further related work}
 
\subsubsection{Accountability: positive results}

The closest analog to our work in the distributed computing literature
is a sequence of papers, including Buterin and
Griffith~\cite{buterin2017casper}, Civit et
al.~\cite{civit2021polygraph}, and Shamis et al.~\cite{shamis2022ia},
about protocols that satisfy \emph{accountability}, meaning protocols
that can provide certificates of guilt in the event of a consistency
violation. 
Further examples of papers in this sequence include Sheng
et al.~\cite{sheng2021bft}, who analyze accountability for well-known
permissioned protocols such as HotStuff~\cite{yin2019hotstuff}, PBFT
\cite{castro1999practical}, Tendermint
\cite{buchman2016tendermint,BKM18}, and Algorand
\cite{chen2018algorand}; and Civit et
al.~\cite{civit2022crime,civit2023easy}, who describe generic
transformations that take any permissioned protocol designed for the
partially synchronous setting
and provide a corresponding accountable version.

None of these papers describe how to carry out asymmetric punishment
(e.g., slashing) and thus fall short of our goals here---that is, they
address challenge~(1) on page~\pageref{page:challenges}, but not
challenges~(2) and~(3).  One exception to this point is the ZLB
protocol of Ranchal-Pedrosa and Gramoli~\cite{ranchal2020zlb}, which
is a permissioned blockchain protocol (with a fixed and known player
set) that is able to implement slashing when the adversary controls
less than a $5/9$ fraction of the player set, but which does not address the issue of ensuring that slashing is implemented before guilty parties are able to trade out of their position (i.e.\ challenge (2)).
Freitas de~Souza et al.~\cite{de2021accountability} also describe a process for removing guilty parties in a protocol for lattice agreement (this abstraction is
weaker than consensus and can be implemented in an asynchronous system), but their protocol does not ensure slashing occurs before guilty parties can trade out of their position, and also assumes an honest majority. 
Sridhar et
al.~\cite{sridhar2023better} specify a ``gadget'' that can be
applied to blockchain protocols operating in the synchronous setting
to reboot and maintain consistency after an attack, but they do not
implement slashing and assume that an honest majority is somehow
reestablished out-of-protocol.

Even more significantly, with the exception of
\cite{sridhar2023better} and
\cite{neu2022availability,tas2022babylon,tas2023bitcoin} (described
below), the entire literature on accountability considers only
permissioned protocols.  Turning a permissioned protocol into a
permissionless one is generally technically challenging (if not
impossible) due to the three additional challenges listed in
Section~\ref{ss:intro}. For example, while Dwork et al.~\cite{DLS88}
showed in 1988 how to achieve consistency and liveness in the
partially synchronous setting with an adversary that controls less
than one-third of the players, the first permissionless analog of this
result was proved only last year \cite{lewis2023permissionless}.  Our
positive result here (Theorem~\ref{posres}) provides a
(permissionless) proof-of-stake protocol that, for an adversary
controlling less than two-thirds of the total stake, provably
implements slashing: It reaches consensus on slashing conditions, even
in the face of consistency violations, before the adversary is able to
remove their stake, thereby guaranteeing the EAAC property.  As the
protocol overview in Section~\ref{six} makes clear, a number of new
ideas are required to obtain this result. Our positive result constitutes a
significant advance even in the permissioned setting (relative
to~\cite{ranchal2020zlb}) in that our protocol can punish adversaries
that control less than two-thirds (as opposed to five-ninths) of the
overall player participation. As noted in Section~\ref{ss:overview},
this bound of $2/3$ is the best possible.

Prior to the study of accountability, Li and Mazieres \cite{li2007beyond} considered how to design BFT protocols that still offer certain guarantees when more than $f$ failures occur. The describe a protocol called BFT2F which has the same liveness and consistency guarantees as PBFT when no more than $f<n/3$ players fail; with more than $f$ but no more than $2f$ failures, BFT2F prohibits malicious players from making up operations that clients have never issued and prevents certain kinds of consistency violations.

\subsubsection{Accountability: negative results}

The literature on accountability has focused primarily on positive
results. One exception is Neu et al.~\cite{neu2022availability}, who
prove that no protocol operating in the dynamically available setting
can provide accountability.
The authors then provide an approach to
addressing this limitation by describing a ``gadget'' that checkpoints
a longest-chain protocol. The ``full ledger'' is then live in the
dynamically available setting, while the checkpointed prefix ledger
provides accountability.  Another exception is Tas et
el.~\cite{tas2022babylon,tas2023bitcoin} who, in addition to positive
results on defending against long-range attacks, prove negative
results on accountability for adversaries that control at least
two-thirds of the player set.
None of these papers provide any economic modeling, which is a key
contribution of the present work (see Section~\ref{costdef}).



\subsubsection{The cost of double-spend attacks}

On the economics side, the most closely related work to ours is that
of Budish~\cite{budish2018economic,budish2023trust}, which is reviewed
briefly in Section~\ref{ss:flow}; see the references therein and
Halaburda et al.~\cite{halaburda2022microeconomics} for a broader view
of economics research on blockchain protocols.
Bonneau~\cite{bonneau16}, Leshno et al.~\cite{LPS23}, and Gans and
Halaburda~\cite{GH23} 
provide
additional arguments that, at least for 
Bitcoin-like protocols, double-spend attacks may be very cheap to execute if there
are no post-attack consequences.  Leshno et al.~\cite{LPS23} also
propose a variant of the Bitcoin protocol that halts whenever a
consistency violation is detected (in effect, making consistency
violations impossible by converting them into liveness violations).

\subsubsection{The game theory of slashing}

Deb et al.~\cite{stakesure}
consider the game theory of slashing from an angle that is
complementary to ours, in a model in which the adversary need not own
resources to carry out an attack, but can instead offer bribes to
players. (Bonneau~\cite{bonneau16} and Leshno et al.~\cite{LPS23}
similarly consider the possibility of bribery by an attacker.)  The
authors argue that, without slashing, rational players can be
incentivized to accept small bribes to deviate from the protocol, even
when such deviations by a large number of players causes a consistency
violation and a collapse in the value of their
resources. Fundamentally, the reason for this is that the price
collapse is non-targeted, while bribe pay-offs depend directly on
individual actions. The authors point out that targeted slashing,
assuming it could be somehow implemented, could ensure that a
unilateral deviation by a single player would lead to significant
punishment by the protocol, in which case bribes would be effective
only if they were large.

\section{The model}  \label{setup} 

Impossibility results such as Theorems~\ref{neg} and~\ref{neg2} require a precise model of the permissionless consensus protocol design space; we adopt the one defined by Lewis-Pye and Roughgarden \cite{lewis2023permissionless}.
In brief:
\begin{itemize}
\item There is a set of players (unknown to the protocol). Each player can control an unbounded number of
 identifiers and at each timeslot could be active or inactive.
 \item Messages are disseminated (e.g., via a gossip protocol) rather than sent point-to-point. 
Disseminated messages are eventually received by their recipients, possibly after a delay.
\item In the synchronous setting, there is a finite upper bound, known to the protocol, on the worst-possible message delay. In the partially synchronous setting, there is an initial period of unknown finite duration in which message delays may be arbitrary.
\item The behavior of a player is a function of the timeslot, its internal state, and the information it has received to-date from other players and from ``oracles'' (described below). Formally, a protocol is a specification of (the intended version of) this function.
\item Each player maintains a running (ordered) list of transactions that it regards as confirmed.
\item Players can own resources, which may be ``external'' to the protocol (e.g., hashrate) or ``on-chain'' (e.g., registered stake). External resources evolve independently of the protocol, while the values of on-chain resources generally depend on the transactions that have been confirmed thus far.
\item A $\rho$-bounded adversary is one that never controls more than a $\rho$ fraction of the overall amount of an (external or on-chain) resource that is controlled by the currently active players.
\item Protocols may make use of ``oracles'' that represent cryptographic primitives. External resources are modeled as a special type of oracle (called a ``permitter'') to which the allowable queries depend on a player's resource balance.
\item A protocol satisfies consistency if players' running transaction
  lists are consistent across players and across time (if $\mathtt{T}$
  and $\mathtt{T}'$ are the lists of honest players $p$ and $p'$ at
  timeslots $t$ and~$t'$, then either $\mathtt{T}$ is a prefix of
  $\mathtt{T}'$ or vice versa).
\item A protocol satisfies liveness if, during periods of synchrony, honest players regularly add new confirmed transactions to their running lists.
\end{itemize}
Sections~\ref{ss:model_first}--\ref{protdef} provide the mathematical details, and further discussion of the model can be found in \cite{lewis2023permissionless}.

In addition to this design space, we adopt the ``hierarchy of permissionlessness'' proposed in \cite{lewis2023permissionless}. Intuitively, and phrased here specifically for proof-of-stake protocols, the two key definitions are the following (see Section~\ref{hierarchy} for details):
\begin{itemize}

\item In the dynamically available setting, at every timeslot, at least one non-malicious player with a non-zero amount of
registered stake is active.

\item In the quasi-permissionless setting, at every timeslot, every non-malicious player with a non-zero amount of registered stake is active.

\end{itemize}

The rest of this section fills in the details of these definitions. The reader interested in attack cost modeling can skip to Section~\ref{costdef}, referring back to this section as needed (e.g., when reading the proofs of Theorems~\ref{neg} and~\ref{neg2}).

\subsection{The set of players and the means of communication}\label{ss:model_first}

\vspace{0.2cm} 
\noindent \textbf{The set of players}. We consider a potentially
infinite set of players $\mathcal{P}$. Each player $p\in \mathcal{P}$
is allocated a non-empty and potentially infinite set of
\emph{identifiers} $\mathtt{id}(p)$.  One can think of
$\mathtt{id}(p)$ as an arbitrarily large pre-generated set of public
keys for which~$p$ knows the corresponding private key; a player~$p$
can use its identifiers to create an 
arbitrarily large number of sybils.
Identifier sets are disjoint, meaning
$\mathtt{id}(p)\cap \mathtt{id}(p')=\emptyset$ when $p\neq p'$; intuitively, no player knows the private keys that are held by other players. 

\vspace{0.2cm} 
\noindent \textbf{Permissionless entry and exit}. Time is divided into
discrete timeslots $t=1,2,\dots$,  and each player may or may not be \emph{active} at each
timeslot.  A \emph{player allocation} is a function specifying
$\mathtt{id}(p)$ for each $p\in \mathcal{P}$ and the timeslots at
which each player is active.  Because protocols generally require
\emph{some} active players to achieve any non-trivial functionality,
 we assume that a non-zero but finite number of players is active at
each timeslot. The player allocation is exogenous, in that we do not
model why a player might be active or inactive at a given timeslot.



\vspace{0.2cm} 
\noindent \textbf{Inputs}. Each player is given a finite set of
\emph{inputs}, which capture its knowledge at the beginning of the
execution of a protocol.  If a variable is specified as one of $p$'s
inputs, we refer to it as \emph{determined for} $p$, otherwise it is
\emph{undetermined for} $p$. If a variable is determined/undetermined
for all $p$ then we refer to it as \emph{determined/undetermined}. For
example, to model a permissionless environment with sybils, 
we assume that
$\mathtt{id}(p)$ and the timeslots at which $p$ is active are
determined for $p$ but undetermined for $p' \neq p$.

\vspace{0.2cm} 
\noindent \textbf{Message sending}.  At each timeslot, each active
player may \emph{disseminate} a finite set of messages (each of finite
length), and will receive a (possibly empty)
multiset of messages that have been disseminated by other players at
previous timeslots.  Inactive players do not disseminate or receive
any messages.  
We consider two of the most common models of communication reliability,
the synchronous and partially synchronous
models. 
These models have to be adapted, however, to deal with the
fact that players may not be active at all timeslots:

\noindent \emph{Synchronous model.}  There exists some determined
$\Delta \in \mathbb{N}_{>0}$ such that if $p$ disseminates $m$ at $t$,
and if $p'\neq p$ is active at $t'\geq t+\Delta$, then $p'$ receives
that dissemination at a timeslot $\leq t'$.

\noindent \emph{Partially synchronous model.}  There exists some
determined $\Delta \in \mathbb{N}_{>0}$, and undetermined timeslot
GST such that, if $p$ disseminates $m$ at $t$ and 
$p'\neq p$ is active at
$t'\geq \text{max} \{ \text{GST}, t \} +\Delta$, then $p'$ receives
that dissemination at a timeslot $\leq t'$.


\subsection{Players and oracles}\label{ss:st}

Using an approach that is common in distributed computing, we model
player behavior via state machine diagrams. The description of a
protocol also specifies a (possibly empty) set of \emph{oracles}
$\mathcal{O}= \{ O_1,\dots, O_z \}$, which are used to capture
idealized cryptographic primitives. Players may send \emph{queries} to
the oracles and will then receive \emph{responses} in return. At each
timeslot $t$, the state transition made by player $p$ therefore
depends on the oracle responses received by $p$ at $t$, as well as
$t$, $p$'s present state, and the messages received by $p$ at $t$.  We
refer the reader to \cite{lewis2023permissionless} for a simple
description of how oracles can be used to model standard cryptographic
primitives such as signature schemes, verifiable delay functions, and
ephemeral keys.  Section \ref{mec} and Appendix \ref{aa:bitcoin}
describe precisely how oracles can be used to model
\emph{external resources} such as ASICs (for proof-of-work protocols)
and memory chips (for proof-of-space protocols).

\vspace{0.2cm}
\noindent \textbf{Byzantine and honest players}. 
To ensure that our impossibility results are as strong
as possible, we consider a \emph{static adversary}. 
In the static adversary model, each player is either \emph{Byzantine}
or \emph{honest} and an arbitrary and undetermined subset of the
players may be Byzantine. The difference between Byzantine and honest
players is that honest players must have the state diagram 
specified by the protocol, while Byzantine players may have arbitrary
state diagrams.\footnote{While one might suppose that honest players are incentivized by inflationary rewards (e.g., block rewards or staking rewards paid out in newly minted coins by the protocol to those that appear to run it honestly) and Byzantine players are motivated by off-chain gains (e.g., profit from a double-spend attack or a judiciously chosen short position), we do not attempt to microfound why a given player might choose to be honest or Byzantine.} To model a perfectly coordinated adversary, we also
allow that the instructions carried out by each Byzantine player can
depend on the messages and responses received by other Byzantine
players.  That is, if $p$ is Byzantine, the 
messages $p$ disseminates, the queries $p$ sends, and $p$'s
state transition at a timeslot~$t$ are a function not only of $p$'s
state and the multiset of messages and oracle responses received by
$p$ at $t$, but also of the corresponding values for all the other
Byzantine players. 

\subsection{Modeling external resources} \label{mec} 
\textbf{External vs.\  on-chain resources}. We suppose blockchain validation requires \emph{resources}, which can either be \emph{external} or \emph{on-chain}. External resources are the hardware (such as ASICs or memory chips) required by validators in proof-of-work (PoW) or proof-of-space (PoSp) protocols. By contrast, on-chain resources are those such as stake that are recorded on the blockchain.  From the perspective of our analysis here, a key distinction between external and on-chain resources is that the latter can be confiscated by consensus amongst validators.

\vspace{0.2cm}
\noindent \textbf{External resources and permitters}.
      \emph{Permitter oracles}, or simply {\em permitters}, are
required for modeling external resouces (but not for on-chain
resources). A protocol may use a finite set of permitters.
These are listed as part of the protocol 
amongst the oracles in~$\mathcal{O}$, but have some distinguished
features. 
  
\vspace{0,2cm} 
\noindent \textbf{The resource allocation}. For each execution of the
protocol, and for each permitter oracle $O$, we are given a
corresponding \emph{resource allocation}, denoted $\mathtt{R}^O$.  We
assume that $\mathtt{R}^O$ is undetermined, as befits  an ``external''
resource.  The resource allocation can be thought of as assigning
each player some amount of the external resource at each
timeslot.  That is, for all $p\in \mathcal{P}$ and $t$, we have
$\mathtt{R}^O(p,t)\in \mathbb{N}$.  We refer to $\mathtt{R}^O(p,t)$
as $p$'s {\em balance} at $t$.  Because  balances are
unknown to a protocol, an inactive player might as well have a zero
 balance: $\mathtt{R}^O(p,t)=0$ whenever $p$ is not active at
$t$.  
For each $t$, we also define
$\mathtt{R}^O(t):= \sum_p \mathtt{R}^O(p,t)$.  

\vspace{0.2cm} 
\noindent \textbf{Restricting the adversary}. An arbitrary value
$\mathtt{R}^O_{\text{max}}$ is given as a protocol input. This value is
determined, but the protocol must function for any given value of
$\mathtt{R}^O_{\text{max}}\geq 1$.\footnote{The upper bound $\mathtt{R}^O_{\text{max}}$ on the total resource balance can be very loose, for example representing all of the silicon on planet earth.
} Let $B$ denote the set of Byzantine players and define
$\mathtt{R}^O_B(t):= \sum_{p\in B} \mathtt{R}^O(p,t)$.  For
$\rho \in [0,1]$, we say that $\mathtt{R}^O$ is {\em $\rho$-bounded}
if the following conditions are satisfied for all $t$:
  \begin{itemize} 
  \item   $\mathtt{R}^O(t)\in [1, \mathtt{R}^O_{\text{max}}]$. 
  \item $\mathtt{R}^O_B(t)/\mathtt{R}^O(t) \leq \rho$.
  \end{itemize}  
The smaller the value of~$\rho$, the more severe the restriction on
the combined ``power'' of the Byzantine players.  
 If all resource allocations corresponding to permitters in~$\mathcal{O}$ are $\rho$-bounded, then we say the adversary is \emph{externally} $\rho$-\emph{bounded}.

\vspace{0,2cm} 
\noindent \textbf{Permitter oracles}. At each timeslot, a player may send \emph{queries} to each permitter oracle. These queries can be thought of as requests for proof-of-work or proof-of-space, and the player will then receive a \emph{response} to each query at the same timeslot. The difference between permitter
oracles and other oracles is that the queries that a player $p$ can
send to a permitter oracle $O$ at timeslot $t$ depend on
$\mathtt{R}^O(p,t)$.  If a player $p$ sends a query to the permitter
oracle $O$ at timeslot $t$, the query must be of the form $(b,\sigma)$ such
that $b\in \mathbb{N}$ and $b\leq \mathtt{R}^O(p,t)$. Importantly,
this constraint applies also to Byzantine players. In the case of a proof-of-work protocol, the  query $(b,\sigma)$ can be thought of as a request for a proof-of-work for the string $\sigma$, to which hashrate $b$  is committed at timeslot $t$. 

\vspace{0,2cm} 
\noindent \textbf{Single and multi-use permitters}. A player may make multiple queries~$(b_1,\sigma_1),\ldots,(b_k,\sigma_k)$ to a
permitter in a single timeslot~$t$, subject to the following constraints.  With a
{\em single-use} permitter---the appropriate version for modeling the
Bitcoin protocol (see Appendix~\ref{aa:bitcoin} for details)---these queries are
restricted to satisfy $\sum_{i=1}^k b_i \le \mathtt{R}^O(p,t)$.
(Interpreting the~$b_i$'s as hashrates devoted to the queries, this
constraint asserts that none of a player's overall hashrate can be
``reused.'')  With a {\em multi-use} permitter---the more convenient
version for modeling  protocols that incorporate proof-of-space such as Chia \cite{cohen2019chia}---the only restriction is
that~$b_i \le \mathtt{R}^O(p,t)$ for each~$i=1,2,\ldots,k$.
(Intuitively, space devoted to storing lookup tables can be reused
across different ``challenges.'')

\vspace{0,2cm} 
\noindent \textbf{Permitter responses}. If $O$ is deterministic, then
when $p$ sends the query $(b,\sigma)$ to $O$ at timeslot $t$, the
values $t$, $b$, and $\sigma$ determine the response $r$ of the permitter.
To prevent responses from being forged by Byzantine players, $r$ is a
message signed by the permitter. If $O$ is probabilistic, then $t$, $b$,
and $\sigma$ determine a distribution on responses.


\subsection{Modeling stake} \label{stake} 


For the sake of simplicity,  we restrict attention here to on-chain resources that are forms of stake. We refer the reader to \cite{lewis2023permissionless} for further discussion of other forms of on-chain resources, and note that the results we present here are easily adapted to accommodate general on-chain resources.

\vspace{0.2cm}   
\noindent \textbf{The approach to modeling stake}.  This section
defines \emph{stake allocation functions}, which take
\emph{transactions} as inputs. Section \ref{bcp} 
describes how \emph{blockchain protocols} are required to select
specific sets of transactions to which stake allocation functions can
be applied so as to specify the stake of each player at a given point
in a protocol execution.

\vspace{0.2cm} 
\noindent \textbf{The environment}.  For each execution of a
blockchain protocol, there exists an undetermined \emph{environment},
denoted $\mathtt{En}$, which sends \emph{transactions}
to players.\footnote{For convenience, in the description of the PosT
  protocol in Section~\ref{posproof}, we also allow players to issue
  and sign special types of transactions, for example to signal the end
  of an ``epoch.''}
Transactions are messages signed by the environment.
 If $\mathtt{En}$ sends $\mathtt{tr}$ to $p$ at timeslot $t$, then $p$
\emph{receives} $\mathtt{tr}$ at $t$ as a member of its multiset of
messages received at that timeslot. Formally, the environment
$\mathtt{En}$ is simply a set of triples of the form
$(p,\mathtt{tr},t)$ such that $p$ is active at $t$. We stipulate that,
if $(p,\mathtt{tr},t)\in \mathtt{En}$, then $p$ receives the
transaction $\mathtt{tr}$ at $t$, in addition to the other
disseminations that it receives at $t$. We assume that, for each
$p\in \mathcal{P}$ and each $t$, there exist at most finitely many
triples $(p,\mathtt{tr},t)$ in $\mathtt{En}$.

  
\vspace{0.2cm} We allow a protocol to specify a finite set of
\emph{stake allocation functions}, representing one or more forms of
stake (e.g., stake amounts held in escrow in a designated
staking contract).


 \vspace{0.2cm} 
 \noindent \textbf{The initial stake distribution corresponding to a stake allocation function}. Corresponding to each stake allocation function $\mathtt{S}$ is an \emph{initial distribution}, denoted $\mathtt{S}^*$, which is given to every player as
an input.
This distribution allocates
a positive integer amount of stake to each of a 
finite and non-zero number of identifiers,
 and can be thought of as chosen by an adversary, subject to any
 constraints imposed on the fraction of stake controlled by Byzantine
 players.

\vspace{0.2cm} 
\noindent \textbf{Stake allocation functions}.    If
$\mathtt{T}$ is a sequence of transactions,
then $\mathtt{S}(\mathtt{S}^*,\mathtt{T}, id)$ ($\in \mathbb{N}$) is the
stake owned by identifier $id$ after execution of the transactions in
$\mathtt{T}$. It will also be notationally convenient to let
$\mathtt{S}(\mathtt{S}^*,\mathtt{T})$ denote the function which on input $id$
gives output $\mathtt{S}(\mathtt{S}^*,\mathtt{T}, id)$.
If $\mathtt{T}$ is a sequence of transactions, then $\mathtt{T} \ast \mathtt{tr}$ denotes the sequence $\mathtt{T}$ concatenated with $\mathtt{tr}$. 

We conclude this section with some baseline assumptions about how
stake works.\footnote{These are important only in the proof of
  Theorem~\ref{neg2}.  They hold for the PosT protocol described in
  Section~\ref{posres}, and would presumably be satisfied by any
  reasonable PoS protocol.}  
We assume that players' initial allocations can be transferred.
Formally, given stake allocation
functions $\{ \mathtt{S}_1,\ldots,\mathtt{S}_j \}$, we assume that:
for all initial distributions~$\mathtt{S}^*_1,\ldots,\mathtt{S}^*_j$
and every finite subset~$I$ of identifiers, there exists a set of
transactions $\mathtt{T}$ such that, no matter how they are ordered,
$\mathtt{S}_h(\mathtt{S}^*_h,\mathtt{T},id) = 0$ for every $h \in [j]$
and $id \in I$.\footnote{For example, if a stake allocation function
  represents native cryptocurrency, $\mathtt{T}$ could comprise
  payments transferring all stake initially owned by identifiers
  in~$I$ to identifiers outside of~$I$. If the stake allocation
  function tracks stake-in-escrow, $\mathtt{T}$ could include one
  unstaking transaction for each $id \in I$.}  We also assume that
some transactions~$\mathtt{tr}$ are {\em benign} in the sense that
they do not destroy this property: for
every~$\mathtt{S}^*_1,\ldots,\mathtt{S}^*_j$ and~$I$, there
exists~$\mathtt{T}$ such that, no matter how the transactions of
$\mathtt{T}$ are ordered,
$\mathtt{S}_h(\mathtt{S}_h^*,\mathtt{tr} \ast \mathtt{T},id) = 0$ for
every $h \in [j]$ and $id \in I$.\footnote{For example, a simple
  payment between two identifiers outside of~$I$ would presumably be a
  benign transaction.}  We do not assume that all transactions are
benign in this sense; for example, the PosT protocol in
Section~\ref{posproof} uses (non-benign) transactions that are
``certificates of guilt'' which have the effect of freezing the
assests of the implicated identifiers.


\subsection{Blockchain protocol requirements} \label{bcp} 

  \vspace{0.2cm}
\noindent \textbf{Confirmed transactions}. Each \emph{blockchain protocol}
specifies a \emph{confirmation rule} $\mathcal{C}$, which
is a function that takes as input an arbitrary set of messages $M$ and
returns a sequence $\mathtt{T} $ of the transactions among
those messages.
At timeslot $t$, if $M$ is
the set of all messages received by an honest player~$p$ at timeslots
$\leq t$, then $p$ regards the transactions in $\mathcal{C}(M)$ as
\emph{confirmed}.
For a set of messages~$M$, 
define $\mathtt{S}(\mathtt{S}^*,M,id):=\mathtt{S}(\mathtt{S}^*,\mathcal{C}(M),id)$ and
$\mathtt{S}(\mathtt{S}^*,M):=\mathtt{S}(\mathtt{S}^*,\mathcal{C}(M))$.


\vspace{0.2cm} The requirements on a blockchain protocol are that it
should be \emph{live} and \emph{consistent}. For the sake of
simplicity (to avoid the discussion of small error probabilities), in
this paper we consider versions of liveness and consistency that apply
to deterministic protocols. The proofs of our negative results
(Theorems~\ref{neg} and~\ref{neg2}) will apply directly to
deterministic protocols, but are easily extended to give analogous
impossibility results for probabilistic protocols, simply by
accounting for the appropriate error probabilities. The proof of our
positive result (Theorem~\ref{posres}) uses 
a deterministic protocol, which only strengthens the result.

\vspace{0.2cm} 
\noindent \textbf{Defining
  liveness}. 
We say a protocol is \emph{live} if there exists a constant $T_l$,
which may depend on the message delay bound $\Delta$ and the other
determined protocol inputs, such that, whenever the environment sends
a transaction $\mathtt{tr}$ to an honest player at some timeslot $t$,
$\mathtt{tr}$ is among the sequence of confirmed transactions for
every active honest player at every timeslot
$t'\geq \text{max} \{ t, \text{GST} \} +T_l$. (In the synchronous
model, GST should be interpreted as 0.)

    \vspace{0.2cm} 
\noindent \textbf{Defining consistency}. Suppose the sequence of confirmed transactions for $p$ at $t$ is $\sigma= (\mathtt{tr}_1,\dots,\mathtt{tr}_k)$, and that the sequence of confirmed transactions for $p'$ at $t'\geq t$ is $\sigma'=(\mathtt{tr}_1',\dots,\mathtt{tr}'_{k'})$.  We say a blockchain protocol is \emph{consistent} if it holds in all executions that, whenever $p$ and $p'$ are honest: 
\begin{itemize} 
\item If $p=p'$ then $\sigma'$ \emph{extends} $\sigma$, meaning that
  $k'\geq k$ and $\mathtt{tr}_i=\mathtt{tr}'_i$ for each $i\in [1,k]$. 
\item Either $\sigma$ extends $\sigma'$, or $\sigma'$ extends $\sigma$. 
\end{itemize}

\subsection{Protocols, executions, and $\rho$-bounded adversaries} \label{protdef} 


\noindent \textbf{Specifying blockchain protocols and executions}. A
blockchain protocol is a tuple
$(\Sigma,\mathcal{O},\mathcal{C},\mathcal{S})$, where $\Sigma$ is the
state machine diagram determining honest players, $ \mathcal{O}$ is a
set of oracles (some of which may be permitters), $\mathcal{C}$ is a
confirmation rule, and
$ \mathcal{S}=\{ \mathtt{S}_1,\dots, \mathtt{S}_j \}$ is a set of
stake allocation functions.  An \emph{execution} of the protocol
$(\Sigma,\mathcal{O},\mathcal{C},\mathcal{S})$ is a specification of
the set of players $\mathcal{P}$, the player allocation, the state
diagram of each player and their inputs, and the following values for
each player $p$ at each timeslot: (i) $p$'s state at the beginning of
the timeslot; (ii) the multiset of messages received by $p$; (iii) the
oracle queries sent by $p$; (iv) the oracle responses received by $p$;
(v) the messages disseminated by $p$.

 
\vspace{0.2cm} 
\noindent \textbf{Defining $\rho$-bounded adversaries.}  We say that an
execution of a  protocol is {\em $\rho$-bounded} if:
\begin{itemize}

\item The adversary is externally $\rho$-bounded (in the sense of Section~\ref{mec}).

\item For each stake allocation function $\mathtt{S}\in \mathcal{S}$, and among active players, Byzantine players never control more than
  a $\rho$ fraction of the stake.  Formally,
for every honest player $p$ at timeslot~$t$,
if~$\mathtt{T}$ is the set of
transactions that are confirmed for $p$ at $t$ in this execution,
then at most a~$\rho$ fraction of the stake allocated to 
players active at $t$ by $\mathtt{S}(\mathtt{S}^*,\mathtt{T})$ is allocated to
Byzantine players.
\end{itemize}
 
\noindent When we say that ``the adversary is $\rho$-bounded,'' we mean that we
restrict attention to $\rho$-bounded executions. We say that a protocol is {\em $\rho$-resilient} when it
it live and consistent  under the assumption that
the adversary is $\rho$-bounded. We say that a protocol is {\em $\rho$-resilient for liveness} when it
it live  under the assumption that
the adversary is $\rho$-bounded, and that it is {\em $\rho$-resilient for consistency} when it
is consistent  under the assumption that
the adversary is $\rho$-bounded.

\subsection{The dynamically available and quasi-permissionless settings} \label{hierarchy}
 
Lewis-Pye and Roughgarden~\cite{lewis2023permissionless} describe  a ``degree of
permissionlessness'' hierarchy that parameterizes what a protocol may assume about the activity of honest players.
This hierarchy is defined by four settings. Informally:
\begin{enumerate}

\item {\em Fully permissionless setting.} At each moment in time, the
  protocol has no knowledge about which players are currently running
  it.  Proof-of-work protocols are typically
  interpreted as operating in this setting.

\item {\em Dynamically available setting.} At each moment in time, the
  protocol is aware of a dynamically evolving list of identifiers
  (e.g., public keys that currently have stake committed in a
  designated staking contract).  The protocol may assume that at least
  \emph{some} honest members of this list are active and participating
  in the protocol, but must function even when levels of participation
  fluctuate unpredictably.  Proof-of-stake longest-chain protocols 
are typically designed to function correctly in this
  setting.

\item {\em Quasi-permissionless setting.} At each moment in time, the
  protocol is aware of a dynamically evolving list of identifiers (as
  in the dynamically available setting), but now the protocol may
  assume that \emph{all} honest members of the list are active.
  Proof-of-stake PBFT-style protocols 
are typically interpreted as
  operating in this setting.

\item {\em Permissioned setting.} The list of identifiers is fixed at
  the time of the protocol's deployment, with one identifier per
  participant and with no dependence on the protocol's execution.  At
  each moment in time, membership in this list is necessary and
  sufficient for current participation in the protocol.  PBFT is a
  canonical example of a blockchain protocol that is designed for the
  permissioned setting.

\end{enumerate}
Each level of the hierarchy is a strictly easier setting (for
possibility results) than the previous 
level. For example, an impossibility result for the
dynamically available setting such as Theorem~\ref{neg} automatically
holds also in the fully permissionless setting.

\vspace{0.2cm} 
\noindent \textbf{Formally defining the dynamically available setting.} In the dynamically available setting, no assumptions are made about
participation by honest players, other than the minimal assumption
that, if any honest player owns a non-zero amount of stake, then at
least one such player is active.\footnote{Additional assumptions about the
fraction of stake controlled by active honest players are
phrased using the notion of $\rho$-bounded adversaries from
Section~\ref{protdef}.} 

Consider the protocol $(\Sigma,\mathcal{O},\mathcal{C}, \mathcal{S})$.  By
definition, an execution of the protocol is {\em
  consistent with the dynamically available setting} if, for each stake allocation function $\mathtt{S}\in 
  \mathcal{S}$:
\begin{itemize}

\item [] Whenever~$p$ is honest and active at~$t$, with~$\mathtt{T}$
  the set of transactions confirmed for~$p$ at $t$ in this execution, if
  there exists an honest player assigned a non-zero amount of stake by
  $\mathtt{S}(\mathtt{S}^*,\mathtt{T})$, then at least one such player is
  active at~$t$. 

\end{itemize}


\vspace{0.2cm} 
\noindent \textbf{Formally defining the quasi-permissionless setting.} Consider the protocol
$(\Sigma,\mathcal{O},\mathcal{C},\mathcal{S})$.
By definition, an execution of the protocol is {\em consistent with
  the quasi-permissionless setting} if, for each stake allocation function $\mathtt{S}\in 
  \mathcal{S}$:
\begin{itemize}

\item [] 
Whenever~$p$ is honest and active at~$t$, with~$\mathtt{T}$
  the set of transactions confirmed for~$p$ at $t$ in this execution, every honest player that is
  assigned a non-zero amount of stake by
  $\mathtt{S}(\mathtt{S}^*,\mathtt{T})$ is
  active at~$t$. 
  

\end{itemize}
Thus, the quasi-permissionless setting insists on activity from every
honest player that possesses any amount of any form of
stake listed in the protocol description.

\section{The cost of an attack} \label{costdef} 


\subsection{An Overly Simplified Attempt}\label{ss:attempt}

Consider an attacker poised to execute an attack, in the form of a consistency violation, on a permissionless consensus protocol. When would carrying out such an attack have negative economic consequences for the attacker, meaning that, ignoring off-chain gains (from double-spends, short positions, etc.), it's ``worse off'' than before? (Ideally, with honest players ``no worse off'' than before.) A first cut might be to track the market value of all the protocol-relevant resources owned by each (honest or Byzantine) player. That is, consider a blockchain protocol that uses~$k$ resources, and let~$\texttt{R}_i(p,t)$ denote the number of units of the $i$th resource (e.g., ASICs or coins) owned by player~$p$ at timeslot~$t$. 
Let $C_i(t)$ denote the per-unit market price of the $i$th resource at timeslot~$t$ and define $p$'s ``net worth'' at timeslot $t$ by
\begin{equation}\label{eq:networth}
V(p,t) := \sum_{i=1}^k \texttt{R}_i(p,t) \cdot C_i(t).
\end{equation}
The ``consequences of an attack'' carried out by a set $B$ of Byzantine players at some timeslot~$t^*$ could then be measured by the value of $B$'s resources immediately before and after the attack:
\begin{equation}\label{eq:Bsimple}
\frac{\sum_{p \in B} V(p,t^*_+)}{\sum_{p \in B} V(p,t^*_-)},
\end{equation}
with lower values of this ratio corresponding to more severe economic consequences of an attack. The idea that ``the honest players~$H$ should be no worse off'' would then translate to the condition that
\begin{equation}\label{eq:Hsimple}
V(p,t^*_+) \ge V(p,t^*_-)
\end{equation}
for every $p \in H$.

We can then consider how the different scenarios laid out in Section~\ref{ss:scorched} would translate to this formalism.

\vspace{.5\baselineskip}

\noindent
\textbf{Scenario 1: status quo.} Suppose the Bitcoin protocol suffers a consistency violation and yet neither of the narratives in Section~\ref{ss:scorched} plays out as expected, with both the cryptocurrency price and the cryptographic hash function used for proof-of-work mining unchanged following the attack. Then the new market price~$C(t^*_+)$ of an ASIC would equal the old price~$C(t^*_-)$ and hence $V(p,t^*_+)$ would equal $V(p,t^*_-)$ for every player~$p$. (In this example, there's only one resource---hashrate---and so we drop the dependence on~$i$.)
Thus, while the condition~\eqref{eq:Hsimple} would hold (which is good), the ratio in~\eqref{eq:Bsimple} would be~1, indicating an  attack without any economic consequences. We will call such an attack {\em cheap}.
\vspace{.5\baselineskip}

\noindent
\textbf{Scenario 2: price collapse (proof-of-work).} Nakamoto's
original narrative posited that a double-spend attack on the Bitcoin
protocol would significantly decrease the USD value of BTC (and hence
of ASICs for Bitcoin mining), an assumption that translates to
$C(t^*_+) \ll C(t^*_-)$ and hence, for any (honest or Byzantine)
ASIC-owning miner~$p$, $V(p,t^*_+) \ll V(p,t^*_-)$. In this case, the
attack is {\em expensive}, meaning that the ratio
in~\eqref{eq:Bsimple} is less than one. It is expensive for the wrong
reasons, however, in the sense that the conditon in~\eqref{eq:Hsimple}
fails. We therefore say that the attack is {\em expensive due to collapse}.

\vspace{.5\baselineskip}

\noindent
\textbf{Scenario 3: a hard fork (proof-of-work).} The second narrative in Section~\ref{ss:scorched}, in which a double-spend attack on a proof-of-work protocol is punished through a hard fork that changes the cryptographic hash function used for proof-of-work mining, is mathematically equivalent to the first, with existing ASICs losing much of their value and so $C(t^*_+) \ll C(t^*_-)$. This again is an example of an attack that is expensive, but expensive due to collapse.

\vspace{.5\baselineskip}

\noindent
\textbf{Scenario 4: price collapse (proof-of-stake).} Consider now a proof-of-stake protocol---be it longest-chain, PBFT-type with slashing, or anything else---in which case $C(t)$ denotes the USD-denominated market price at time~$t$ of one coin of the protocol's native cryptocurrency. If a double-spend attack harms this cryptocurrency's price, then, as in Scenario~2, the attack is expensive due to collapse. 

\vspace{.5\baselineskip}

\noindent
\textbf{Scenario 5: successful slashing.} Can any protocol---a proof-of-stake protocol, say---make attacks expensive for the right reasons, not due to collapse? To make this question precise, let's assume that a double-spend attack has no effect on the price of the protocol's native currency, i.e., $C(t^*_+)=C(t^*_-)$. Suppose further that a protocol is able to fulfill the promises of slashing outlined in Section~\ref{ss:scorched}, identifying (at least some of) the Byzantine players and surgically confiscating their stake (without inadvertently destroying any stake owned by honest players). Then, $\texttt{R}(p,t^*_+)$ would be less than $\texttt{R}(p,t^*_-)$ (and hence $V(p,t^*_+)$ would be less than $V(p,t^*_-)$) for at least some $p \in B$, while $V(p,t^*_+)=V(p,t^*_-)$ for all $p \in H$. As a result, the ratio in~\eqref{eq:Bsimple} would be less than~1 even as the condition in~\eqref{eq:Hsimple} holds. In this case, we call the attack {\em expensive even in the absence of collapse}. A protocol would then be called {\em EAAC} if every possible consistency violation would be expensive in this sense.


\subsection{The Formal Definitions}\label{ss:eaac}

\paragraph{Delayed economic consequences.}
The definitions proposed in Section~\ref{ss:attempt} convey the spirit
of our approach, but they are inadequate for a number of reasons.  For
example, there is no hope of designing a protocol that is EAAC in the
sense above: the expression in~\eqref{eq:Bsimple} considers only the
immediate economic cost suffered by an attack, while any
protocol-inflicted punishment would require some time to take
effect. For example, to enact slashing in a proof-of-stake protocol,
honest players need time to communicate, confirm evidence of a
consistency violation, and carry out the appropriate slashing
instructions. (Changes in the price of the protocol's native currency,
should they occur, would presumably also play out over a period of
time.)  Thus, we must instead insist that Byzantine players suffer
economic consequences from a consistency violation at some timeslot
$t_f$ after the timeslot~$t^*$ in which the consistency violation is
seen by honest players (here the ``$f$'' stands for
``final''). Intuitively, $t_f$ is a timeslot by which the aftermath of
the attack has played out, with the currency price re-stabilized
(possibly at a new level) and any protocol-inflicted punishment
complete.

\paragraph{Investment functions.}
Next, the definition of the functions~$\mathtt{R}_i(p,t)$, henceforth
called {\em investment functions}, requires more care.  For an
external resource~$i$, $\mathtt{R}_i(p,t)$ denotes the (objective)
number of units of that resource owned by~$p$ at~$t$, as before.  For
a resource~$i$ that corresponds to some form of stake or other
on-chain resource, $\mathtt{R}_i(p,t)$ may depend on the current state
of the blockchain protocol that controls the resource, meaning the
sets of messages that have been received by each player by
timeslot~$t$.  (We generally abuse notation and continue to write
$\mathtt{R}_i(p,t)$, suppressing the dependence on the sets of
received messages and the protocol that generates them.) We allow
flexibility in exactly how investment functions are defined, but the
rough idea is that $\mathtt{R}_i(p,t)$ should update once a
resource-changing transaction has been ``sufficiently processed'' by
the protocol (e.g., confirmed by an appropriate honest player).  Our
three main results concern investment functions of the following
types:
\begin{itemize}

\item Our negative result for the
dynamically available setting (Theorem~\ref{neg}) holds for any
choice of investment functions, and thus requires no further modeling.

\item
Our negative result for the partially synchronous setting
(Theorem~\ref{neg2}) holds assuming only that, provided no consistency
violation has been seen, zeroing out one's stake
balance eventually zeroes out the corresponding investment function
(i.e., the user eventually recoups their investment).

Formally, consider an investment function $\mathtt{R}_i$
that corresponds to a stake allocation function~$\mathtt{S}$.
Suppose that, in any execution~$\mathtt{E}$ in which:
\begin{itemize}
\item $t$ is a timeslot $\ge$GST;
\item no honest player has seen a consistency violation by timeslot~$t+\Gamma$;
\item for some player~$p$ and honest player~$p'$,
  $\mathtt{S}(\mathtt{S}^*,\mathtt{T}_{t'},id) = 0$ for all $id \in
\mathtt{id}(p)$ and $t' \in [ t,t+\Gamma ]$,
where $\mathtt{T}_{t'}$ denotes the transactions confirmed for~$p'$ at
timeslot~$t'$;
\end{itemize}
it holds that $\mathtt{R}_i(p,t+\Gamma)=0$ and, moreover, that
$\mathtt{R}_i(p,t+\Gamma)=0$ also in every execution that is
indistinguishable from~$\mathtt{E}$ for~$p'$ up to timeslot~$t+\Gamma$
(i.e., every execution in which~$p'$ receives the same inputs,
messages, and oracle responses at each timeslot $<t+\Gamma$ as
in~$\mathtt{E}$).  In this case, we call $\mathtt{R}_i$ a {\em
  $\Gamma$-liquid investment function}. (Intuitively, barring a
consistency violation seen by honest players, $p$ can eventually recoup
its investment from~$p'$, or more generally from any third party that
uses~$p'$ as its ``source of truth'' about which transactions have
been confirmed.)


\item For our positive result in
Theorem~\ref{posres}, we use what we call a {\em canonical PoS}
investment function. Here, there is a single resource, representing
the amount of stake that a player has locked up in escrow.
The resource balance of a
player changes only through staking and unstaking transactions. A
canonical investment function~$R(p,t)$ increases by~$x$ after a
valid staking transaction (with staking amount~$x$) is first confirmed by
some honest player (possibly after a delay), and decreases by~$x$
after a valid unstaking transaction is first confirmed by some
honest player (again, possibly after a delay).
Such a function is~$\Gamma$-liquid, where~$\Gamma$ is the maximum
delay before the investment function reflects a newly confirmed
staking or unstaking transaction.

\end{itemize}
%
In all these cases, we can interpret the economic investment
of a player~$p$ at some timeslot $t$ by the expression
in~\eqref{eq:networth}, as before.\footnote{As an aside, the ``flow
  cost of an attack'' discussed in Section~\ref{ss:flow} would then be
  defined as
  $\sum_{p \in B} \sum_{i=1}^k \sum_{t=1}^{t_f} c_i \cdot
  \mathtt{R}_i(p,t)$,
  where~$B$ denotes the attacking players and the~$c_i$'s are defined
  as in that section.}



\paragraph{Valuation functions.}
Finally, we must specify the allowable behavior of a value function
that represents the economic value of on-chain resources
following a consistency violation.
In general, we allow a quite abstract notion of a \emph{valuation
  function} $v$ that is a
nonnegative function of a set~$P$ of players, a timeslot~$t$, player investments
at that time (the $\mathtt{R}_i(p,t)$'s), market prices at that time
(the $C_i(t)$'s), and the sets of messages that have been received by
each player by that time.  
Our impossibility results (Theorems~\ref{neg}
and~\ref{neg2}) apply to all such valuations, and thus require no
further modeling.

For our positive result (Theorem~\ref{posres}), we use what we call a
{\em canonical PoS} valuation function, which imposes two much
stronger restrictions (thereby strengthening the result).
First, we require it to be {\em
  consistency-respecting}, in the sense that it is defined only for
timeslots~$t$ at which consistency holds (i.e., if~$\mathtt{T}$ and
$\mathtt{T}'$ are the confirmed transactions for honest~$p$ and~$p'$
at~$t$, then one of $\mathtt{T},\mathtt{T}'$ is a prefix of the
other).\footnote{Without this requirement, a protocol could achieve
  the EAAC property by addressing only the first two challenges on
  page~\pageref{page:challenges} and not the third. For example, the
  valuation function could be defined to discount any stake implicated
  by some certificate of guilt known to some honest player, even while
  honest players disagree on what the ``post-slashing'' state should
  be.}  Second, at timeslots at which consistency does hold, the
valuation function equals the value of the as-yet-unslashed
stake-in-escrow of the players in question (at the current market
prices).\footnote{See Section~\ref{instant} for the precise
  definition used in the proof of Theorem~\ref{posres}.}




\paragraph{EAAC protocols.}
We now proceed to our formal definition of an ``EAAC'' protocol. As
noted in the last scenario of Section~\ref{ss:attempt}, to have any
hope of avoiding collateral damage to honest participants, we must
assume that a consistency violation does not affect the market prices
of the relevant resources. Thus, in the following definition, we
consider a fixed price~$C_i$ for each resource rather than an
arbitrary price sequence $\{ C_i(t) \}_{t \ge 1}$.



\begin{definition} \label{d:eaac}
A protocol is {\em EAAC with respect to investment functions~$\{ \mathtt{R}_i(\cdot,\cdot)
  \}_{i \in [k]}$ and valuation function $v$} in a given setting if,
for every choice~$\{ C_i \}_{i \in [k]}$ of fixed resource prices,
every external resource~$i$, and every choice~$\{ \mathtt{R}_i(p,t)
\}_{p \in P, t \ge 1}$ of players' investments in~$i$:
\begin{itemize}

\item [(a)] for every execution of the protocol consistent with the setting and every timeslot~$t \ge 1$,
\begin{equation}\label{eq:eaac1}
v(H,t,\{ \mathtt{R}_i(p,t) \}_{p \in H,i \in [k]},\{ C_i \}_{i \in [k]}, \mathcal{M})
\ge 
\sum_{p \in H} \sum_{i=1}^k \mathtt{R}_i(p,t) \cdot C_i,
\end{equation}
where~$H$ denotes the set of honest players and~$\mathcal{M}$ the sets of messages received by each player by timeslot~$t$ in the execution; and

\item [(b)] for every execution of the protocol consistent with the setting with a consistency violation that is first seen by honest players at a timeslot~$t^*$, there exists a timeslot~$t_f \ge t^*$ such that:
\begin{equation}\label{eq:eaac2}
v(B,t_f,\{ \mathtt{R}_i(p,t_f) \}_{p \in B, i \in [k]},\{ C_i \}_{i \in [k]}, \mathcal{M})
<
\sum_{p \in B} \sum_{i=1}^k \mathtt{R}_i(p,t_f) \cdot C_i,
\end{equation}
where~$B$ denotes the set of Byzantine players and~$\mathcal{M}$ the
sets of messages received by each player by timeslot~$t$ in the
execution.\footnote{For the purposes of
  inequalities~\eqref{eq:eaac1} and~\eqref{eq:eaac2}, we interpret an
  undefined valuation function as $+\infty$.}

\end{itemize}
\end{definition}
Intuitively, part~(a) of Definition~\ref{d:eaac} requires that honest players
can always recoup whatever they may have invested (e.g., they
are never slashed in a proof-of-stake protocol). Part~(b) asserts that
an attacker will be unable to fully cash out, with some of their
investment lost to the protocol. 


For brevity, we sometimes denote by~$\alpha_H$ the ratio between the
left- and right-hand sides of~\eqref{eq:eaac1},
and by~$\alpha_B$ the
analogous ratio for~\eqref{eq:eaac2}. 
(If the right-hand side is 0 or
the left-hand side is undefined, 
we can interpret the ratio as~1.) Thus, the EAAC condition states that
$\alpha_H \ge 1$ should always hold (i.e., if attacks are expensive,
it's for the right reasons, not due to collapse) and that
$\alpha_B < 1$ should hold at some point after a consistency violation
(i.e., attacks are indeed expensive).
If Definition~\ref{d:eaac} holds with condition~(b) always satisfied
with some $\alpha_B \le \alpha < 1$, then we say that the protocol is
{\em $\alpha$-EAAC} 
(with respect to $\{ \mathtt{R}_i(\cdot,\cdot)  \}_{i \in [k]}$
and~$v$).\footnote{While ``1'' is arguably
  the most natural threshold when using these ratios to define
  ``expensive'' and ``collapse,'' a different constant~$x \in (0,1)$
  could be used instead. Our impossibility results would then hold,
  with essentially the same proofs, for
  any choice of $x$.}

\section{Impossibility results} \label{res1} 

Our first result establishes that blockchain protocols that are live
and consistent in the dynamically available setting are always cheaply
attackable in the absence of collapse. In other words, once an
adversary is large enough to cause consistency violations, they are
also large enough to avoid asymmetric punishment.
The proof of Theorem~\ref{neg}
appears in Section \ref{negproof}.

\begin{theorem}[Impossibility Result for the
  Dynamically Available Setting]\label{neg} 
In the dynamically available setting, with a $\tfrac{1}{2}$-bounded
adversary, for every choice of investment functions and valuation
function, no protocol can be live and EAAC.
This holds even in the synchronous model and with 
Byzantine players that have fixed (i.e., time-invariant) resource
balances.
\end{theorem} 


Recall the hierarchy of settings described in Section
\ref{hierarchy}. Theorem \ref{neg} establishes that, if we want to
work with protocols that are non-trivially EAAC in some level of this
hierarchy, then we cannot work in the fully permissionless or
dynamically available settings and must instead focus on the
quasi-permissionless or permissioned settings. 

The next theorem shows
that, even should we work in the quasi-permissionless setting,
protocols cannot be non-trivially EAAC if we work in the partial
synchrony model. 
The proof of Theorem~\ref{neg2} appears in Section~\ref{neg2proof}.


\begin{theorem}[Impossibility Result for the
  Partially Synchronous Setting] \label{neg2} 
In the quasi-permissionless setting and the partial synchrony model,
with a $1-2\rho_l$-bounded
adversary, for every choice of liquid investment functions and valuation
function, no protocol can be $\rho_l$-resilient for liveness and EAAC.
\end{theorem}
For example, suppose we restrict attention to the standard case of
protocols which have the same resilience for liveness and
consistency. 
The seminal result of Dwork, Lynch and
Stockmeyer~\cite{DLS88} establishes that protocols for the partial
synchrony model can be $\rho$-resilient only for $\rho<1/3$. Theorem
\ref{neg2} establishes that, if the adversary can own at least 1/3 of
the resources, then the protocol cannot be EAAC. In fact, the proof
establishes a stronger result: a $\tfrac{1}{3}$-bounded adversary can
cause consistency violations without any punishment, asymmetric or
otherwise (i.e., can carry out attacks that are cheap even
when there is a collapse in the
value of resources after the consistency violation is seen).

\section{Provable slashing guarantees: a possibility result} \label{pos}

Theorem \ref{neg2} might seem to end the quest
for PoS protocols with `slashing' in the partial synchrony model. To
circumvent this difficulty, we can consider protocols that are live
and consistent in partial synchrony (for some small $\Delta$) so long
as the adversary is $\rho$-bounded for $\rho<1/3$, and which,
furthermore, implement slashing when the adversary owns more than 1/3
of the resources but message delivery prior to GST always occurs
within some large known time bound $\Delta^*$ (\emph{which may not be
  $O(\Delta)$}).  We argue that this is a realistic setting: while
blockchain protocols are commonly expected to be live and consistent
in the partial synchrony model with a liveness parameter $T_l$
determined by some small value of $\Delta$ (of the order of 1 second,
say), it may also be reasonable to assume that messages will always
eventually be delivered, either via the communication network or some
``out-of-band'' mechanism, within some sufficiently large time bound
(of the order of a day or a week, say).\footnote{We note also that our
  proof of Theorem \ref{posres} 
only requires that message delays prior to GST are
  bounded by $\Delta^*$ for a limited period
  around the time of an attack. This observation is made precise in
  Section \ref{posproof}.\label{foot:limited}}
Theorem~\ref{neg2} does not rule out non-trivial EAAC protocols in
this setting, provided that the time required for an attacker to recoup
its investment following an attack scales with the worst-case
delay~$\Delta^*$; see also the discussion at the end of
Section~\ref{neg2proof}.

\vspace{0.2cm} Formally, we say a blockchain protocol is
\emph{$(\rho,\rho^*)$-EAAC} 
with respect to investment functions $\{ \mathtt{R}_i(\cdot,\cdot)
\}_{i \in [k]}$ and~valuation function~$v$
if it satisfies the following two
conditions, given any determined values of $\Delta$ and
$\Delta^*\geq \Delta$, and where $\rho^*$ is determined:
\begin{enumerate}

\item[(i)] The protocol is live and consistent for the partial
  synchrony model with respect to $\Delta$, so long as the adversary
  is $\rho$-bounded.  Here, we allow (as in Section \ref{bcp}) 
  the liveness parameter~$T_l$ to depend on $\Delta$ and other
  determined inputs, but require it to be independent of
  $\Delta^*$. \emph{In particular,~$\Delta^*$ may be
    much larger than $T_l$}.

\item[(ii)] The protocol is EAAC 
(with respect to $\{ \mathtt{R}_i(\cdot,\cdot)  \}_{i \in [k]}$
and~$v$) so long as the adversary is
  $\rho^*$-bounded and
  message delays prior to GST are always at most $\Delta^*$ (i.e., if
  $p$ disseminates $m$ at~$t$, and if $p'\neq p$ is active at
  $t'\geq t +\Delta^*$, then $p'$ receives that dissemination at a
  timeslot $\leq t'$).\footnote{The fact that $\Delta^*$ may
be much larger than $T_l$ means that a protocol aiming to be
$(\rho,\rho^*)$-EAAC cannot ignore~$\Delta$ and
simply assume that message delays will always be bounded by
$\Delta^*$.}

\end{enumerate}
If in~(ii) the protocol is~$\alpha$-EAAC in the sense of
Section~\ref{ss:eaac}, then we say that the protocol is
$(\alpha,\rho,\rho^*)$-EAAC
(with respect to $\{ \mathtt{R}_i(\cdot,\cdot)  \}_{i \in [k]}$
and~$v$).

\vspace{.75\baselineskip}

Our main positive result is the following; we provide the proof 
in Section~\ref{posproof}.
Canonical PoS investment and valuation functions are defined
informally in Section~\ref{ss:eaac}, and formally in
Section~\ref{posproof}.  


\begin{theorem}[Non-Trivial EAAC Protocols in the
  Quasi-Permissionless Setting]\label{posres} 
For every $\rho<1/3$ and $\rho^*<2/3$, there exists a PoS
  protocol for the quasi-permissionless setting that is
  $(\alpha,\rho,\rho^*)$-EAAC with respect to a canonical PoS
  investment function and a canonical PoS valuation, where
\[
\alpha = \max\{0, (\rho^*-\tfrac{1}{3})/\rho^* \}.
\]
\end{theorem}

This result is optimal in several senses.
First, an easy adaptation of the classic proof of Dwork, Lynch and
Stockmeyer~\cite{DLS88} shows that protocols cannot be
$(\rho,\rho^*)$-EAAC for $\rho\geq 1/3$. 
Second, the bound of 2/3 on $\rho^*$  is tight: an argument of Tas et
al.~\cite{tas2023bitcoin}, when translated to our framework, implies
that,
even in the synchronous setting, a protocol that is $\rho_l$-resilient
for liveness cannot be EAAC with a $(1-\rho_l)$-bounded adversary. The 2/3 bound in the statement of the theorem can therefore only be improved by weakening liveness requirements (as already mentioned in Section \ref{ss:overview}, we will describe how our protocol can be so modified in Section \ref{posproof}). 
Third, the bound on~$\alpha$ cannot be improved: an adversary with
a~$\rho^*$ fraction of the overall resources can cause a consistency
violation using only one-third of the overall resources in a dishonest
way~\cite{DLS88}; because honest players cannot be punished, the adversary can
guarantee that it retains a $(\rho^*-\tfrac{1}{3})$ fraction of
the overall original resources following its attack.

\section{Overview of the PosT protocol and the proof of Theorem \ref{posres}}\label{six}

We describe a PoS version of Tendermint, which we refer to as PosT. A
quick review of permissioned Tendermint is given in Section
\ref{posproof}.  To specify PosT, we envisage
that players can add and remove stake from a special staking contract,
allowing them to act as `validators'. While added to the staking
contract, stake is regarded as being `in escrow' and cannot be
transferred between players. Removal of stake from escrow is subject
to a delay (of time more than $\Delta^*$).  The stake
allocation function $\mathtt{S}$ indicates the stake that
an identifier has in escrow and which has not been marked for removal
from escrow. This is the balance that describes an identifier's
\emph{weight} as a validator: we'll refer to an identifier's
$\mathtt{S}$-balance as their \emph{validating stake}.
Under the assumptions of the quasi-permissionless setting (see
Section~\ref{hierarchy}), all honest players that possess a non-zero
amount of validating stake at some timeslot are active at that
timeslot.



  \vspace{0.2cm}
  \noindent \textbf{The use of epochs}. Recall that the instructions for Tendermint are divided into \emph{views}. In each view, a designated \emph{leader}  proposes the next block of transactions, and other validators may vote on the proposal to form (one or two) \emph{quorum certificates} (QCs) for the block. The instructions for PosT are divided into \emph{epochs}, where each epoch is a sequence of consecutive views.  Removal of stake from escrow is achieved via confirmation of an appropriate transaction. If this transaction is confirmed by the start of epoch $e$, then the corresponding stake remains in escrow  until the end of the epoch, but no longer contributes to the weight of the corresponding identifier as a validator (their $\mathtt{S}$-balance). It is only at the end of each epoch that the $\mathtt{S}$-balance of an identifier can change, meaning that the set of validators is fixed during each epoch. 
  
\vspace{0.2cm}
\noindent \textbf{Certificates of guilt}. Recall that $\rho^*<2/3$ is
determined and that our protocol is required to be EAAC only so long
as the adversary is $\rho^*$-bounded.  We use standard methods (e.g.,
see \cite{sheng2021bft}) to ensure that any consistency violation
causes the production of a \emph{certificate of guilt}, which is a set
of signed messages proving Byzantine action on the part of some of the
validating stake during a certain epoch. If a consistency violation
occurs in some least epoch $e$, then the fact that $\rho^*<2/3$
becomes crucial. This means that honest stake must participate in the
confirmation of each block while in epoch $e$. If message delay prior
to GST is always at most $\Delta^*$, then honest validators will end
epoch $e$ within time $\Delta^*$ of each other and, by modifying
Tendermint to utilize three stages of voting rather than
two,\footnote{HotStuff \cite{yin2019hotstuff} also makes use of three
  stages of voting in each view but, as explained in Section
  \ref{furtherc}, the motivation in that case is rather different.} we
will be able to argue that a certificate of guilt must be received by
all active honest players within time $2\Delta^*$ of any honest player
entering epoch $e+1$. Defining epochs to be of sufficient length
therefore ensures that any consistency violation in epoch $e$ will
produce a certificate of guilt which is seen by all active honest
players before the end of epoch $e+1$, and before the stake used to
create the consistency violation is removed from escrow.

\vspace{0.2cm}
\noindent \textbf{The recovery procedure}.  
The remaining challenge is to design a recovery procedure for honest
players to use, after receiving a
certificate of guilt, to reach consensus on a new state in which
slashing has been carried out.
In more detail, suppose a consistency
violation occurs in some least epoch $e$, with
all active honest players
receiving a certificate of guilt before the end of epoch $e+1$. 
The goal of the recovery procedure is for all honest
validators (for epoch $e$) to output some common block $b$---an
updated genesis block, in effect---which
\emph{includes all transactions confirmed by the end of epoch $e-1$},
and which `slashes' (renders unspendable) at least 1/3 of the
validating stake for epoch~$e$ (with certificates of guilt provided
for all slashed stake). 

One potential difficulty in implementing such a recovery procedure is
that a consistency violation is only guaranteed to produce a
certificate of guilt for 1/3 of the validating stake. If the adversary
possesses $\geq 5/9$ of the validating stake, then they may possess at
least a fraction $(5/9-1/3)/(2/3)=1/3$ of the validating stake that
remains after slashing.  If one were to naively employ some
Tendermint-like protocol to reach consensus on an updated genesis block
(including slashing conditions) for the next iteration of the
protocol, such an adversary could threaten liveness.

The simplest approach to address this issue, and the one we follow
here given our focus on fundamental possibility and impossibility
results (as opposed to more fine-grained performance considerations),
is to base the recovery procedure instead on
the classic protocol of Dolev and
Strong~\cite{dolev1983authenticated}, with delay $\Delta^*$ between
each round of the protocol. In a first instance of the Dolev-Strong
protocol, a designated \emph{leader} is asked to propose an updated
genesis block. If this instance results in consensus amongst honest
players on such a block~$b$, then the honest majority of validating
stake that remains after slashing (since $\rho^*<2/3$) can ``sign
off'' on this value, producing a form of \emph{certificate} that acts
as proof that $b$ can be used as an updated genesis block.
If not, a
second leader for another instance of the Dolev-Strong protocol is
then asked to propose an updated genesis block, and so on. Whenever an
honest leader is chosen (if not before), the corresponding instance of
the Dolev-Strong protocol will conclude with consensus on an updated
genesis block, from which the main protocol can then
resume.\footnote{For simplicity, our protocol description and analysis
  conclude with the successful agreement by honest players on a
  post-slashing state following a consistency violation. This 
  suffices to establish the EAAC property and prove
  Theorem~\ref{posres}. More generally, the protocol could be run
  repeatedly, triggering the recovery procedure as needed to punish
a consistency violation and produce a new genesis blocks for the next
instance of the protocol.}

The full details of the PosT protocol and the proof of
Theorem~\ref{posres} are provided in Section~\ref{posproof}.



\section{Impossibility result for the dynamically available setting: The proof of Theorem \ref{neg}} \label{negproof}

It suffices to prove the result for the synchronous model.  
(The following proof will be valid even if~$\Delta=1$.)
Suppose
the protocol $(\Sigma,\mathcal{O},\mathcal{C},\mathcal{S})$ is live
(with liveness parameter~$T_l$)
and consistent in the dynamically available setting.

\vspace{0.2cm} 
\noindent \textbf{The intuition}. Consider two disjoint sets of
players $X$ and $Y$ that each own an equal amount of
resources. Consider first an execution of the protocol in which
players in $X$ are Byzantine, while players in $Y$ are honest.
Players in $X$ do not initially communicate with players in $Y$, but
rather \emph{simulate} between them an execution in which they are the
only ones active. Because the protocol is live in the dynamically
available setting, this simulated execution must produce a non-empty
sequence $\mathtt{T}$ of confirmed transactions (even without
contribution from the players in $Y$): Note that this conclusion would
not hold if operating in the quasi-permissionless or permissioned
settings, because the protocol might then require active participation
from players owning a majority of resources to confirm transactions.
Meanwhile, and for the same reason, the honest players in $Y$ must
eventually confirm a sequence of transaction $\mathtt{T}'$ that may be
incompatible with $\mathtt{T}$.

Now suppose that, at some later point, players in $X$ disseminate all
the messages that they would have disseminated if honest in their
simulated execution. At this point, it is not possible for late
arriving players to determine whether the players in $X$ or the
players in $Y$ are honest. If $\alpha_B<x$, then $\alpha_H<x$ in a
symmetrical execution, in which it is the players in $X$ who are
honest, while the players in $Y$ are Byzantine and run their own
simulated execution. (The notation~$\alpha_H$ and~$\alpha_B$ is
defined in the discussion following Definition~\ref{d:eaac}.)

\vspace{0.2cm}
\noindent \textbf{The formal proof}.  We consider two non-empty sets
of players, $X$ and $Y$ of equal size, and four protocol
executions in which all parameters remain the same unless explicitly
stated otherwise. For the sake of simplicity, we suppose all active players
in $X\cup Y$ always hold a single unit of each external resource (if any),
and all players in~$X \cup Y$ begin with one unit of each form of
stake (if any).
If $\mathcal{S}$ is non-empty, then we suppose
that, for each $\mathtt{S}\in \mathcal{S}$, the transactions
$\mathtt{tr}_1$ and~$\mathtt{tr}_2$ mentioned below do not affect the
stake controlled by players in $X \cup Y$ (e.g., they are simple
transfers between players outside of $X \cup Y$).
If $\mathcal{S}$ is empty, then
the form of these transactions is of no significance.

\vspace{0.1cm} 
\noindent \textbf{Execution 1}.  The only active players are those in
$X$, who are active at all timeslots. The environment sends a single
transaction $\mathtt{tr}_1$ to one of the players, $p_1$ say, at
timeslot 1, and does not send any further transactions.  All players
are honest. By liveness, $\mathtt{tr}_1$ is confirmed for all active
players by timeslot $1+T_l$.

\vspace{0.1cm}
\noindent \textbf{Execution 2}.  The only active players are those in
$Y$, who are active at all timeslots. The environment sends a single
transaction $\mathtt{tr}_2$ (with $\mathtt{tr}_2\neq \mathtt{tr}_1$)
to one of the players, $p_2$ say, at timeslot 1, and does not send any
further transactions. All players are honest. By liveness,
$\mathtt{tr}_2$ is confirmed for all active players by timeslot
$1+T_l$.

\vspace{0.1cm} 
\noindent \textbf{Execution 3}. The active players are $X\cup Y$, and
those players are active at all timeslots.  Choose $t^*>1+T_l$. The
environment sends $\mathtt{tr}_1$ to $p_1$ and $\mathtt{tr}_2$ to
$p_2$ at timeslot 1, and does not send any further
transactions. Players in $X$ are Byzantine, ignore messages from
players in $Y$ until~$t^*$, and simulate the players in Execution 1
precisely (carrying out instructions as if they receive precisely the
same messages at the same timeslots), except that they delay the
dissemination of all messages until $t^*-1$.  At $t^*-1$, players in
$X$ disseminate all messages that the players in Execution 1
disseminate at timeslots $<t^*$, and these messages are received by
all active players by timeslot $t^*$. At timeslots $\geq t^*$ players
in $X$ act precisely like honest players, except that they pretend all
messages received from players in $Y$ at timeslots $<t^*$ were
received at $t^*$.  Players in $Y$ are honest.

\vspace{0.1cm} 
\noindent \textbf{Execution 4}. This is the same as Execution 3, but
with the roles of $X$ and $Y$ reversed. The set of active players is
$X\cup Y$, and those players are active at all timeslots.  Timeslot
$t^*$ is defined as before. Again, the environment sends
$\mathtt{tr}_1$ to $p_1$ and $\mathtt{tr}_2$ to $p_2$ at timeslot 1,
and does not send any further transactions. Now players in $Y$ are
Byzantine, ignore messages from players in $X$ until~$t^*$, and
simulate the players in Execution 2 precisely (carrying out
instructions as if they receive precisely the same messages at the
same timeslots), except that they delay the dissemination of all
messages until $t^*-1$. At $t^*-1$, players in $Y$ disseminate all
messages that the players in Execution~2 disseminate at timeslots
$<t^*$. At timeslots $\geq t^*$ players in $Y$ act precisely like
honest players, except that they pretend all messages received from
players in $X$ at timeslots $<t^*$ were received at $t^*$.  Players in
$Y$ are honest.

\vspace{0.2cm} 
\noindent \textbf{Analysis}. We first prove that at least one of
Executions 3 and 4 must see a consistency violation.  To see this,
suppose towards a contradiction that Execution 3 does not. Note that,
until~$t^*$, Execution~3 is indistinguishable from Execution 2 as far
as the honest players in $Y$ are concerned, i.e. they receive
precisely the same inputs, messages and oracle responses at each
timeslot $<t^*$. It must therefore be the case that all players in $Y$
regard $\mathtt{tr}_2$ as confirmed and $\mathtt{tr}_1$ as not
confirmed (because it has not yet been received by those players) by
timeslot $t^*-1$. Let $M^{\ast}$ be the set of messages received by
all honest players by timeslot $t^*$. If there is no consistency
violation then it must be the case that $\mathcal{C}(M^{\ast})$ is a
sequence in which $\mathtt{tr}_1$ does not precede $\mathtt{tr}_2$.

In this case, consider Execution 4. Note that, until $t^*$, Execution
4 is indistinguishable from Execution 1 as far as the honest players
in $X$ are concerned, i.e. they receive precisely the same inputs,
messages and oracle responses at each timeslot $<t^*$. It must
therefore be the case that all players in $X$ regard $\mathtt{tr}_1$
as confirmed and $\mathtt{tr}_2$ as not confirmed by timeslot
$t^*-1$. Note next, however, that the set of messages received by all
honest players by timeslot $t^*$ is precisely the same set $M^{\ast}$
in Executions 3 and 4. We concluded above that $\mathtt{tr}_1$ does
not precede $\mathtt{tr}_2$ in $\mathcal{C}(M^{\ast})$, meaning that
Execution 4 sees a consistency violation at $t^*$.

To complete the proof, without loss of generality, suppose Execution 3
sees a consistency violation; the case that Execution 4 sees a
consistency violation is symmetric. Towards a contradiction, suppose the
protocol is EAAC with respect to the 
investment functions~$\{
\mathtt{R}_i(\cdot,\cdot)
  \}_{i \in [k]}$ and
valuation function~$v$. In that
case, we must have $ \alpha_B<1$ at some timeslot $t_f > t^*$ in
Execution 3 (and~$\alpha_H \ge 1$ for all~$t$).  In that case,
however, consider Execution~4.  
Because the set of messages received by each player by time~$t_f$ is
the same in both executions, $R_i(p,t_f)$ is also the same for
every~$i \in [k]$ and~$p \in X \cup Y$. It follows that the
valuation~$v$ is also the same (for both $X$ and $Y$) at time~$t_f$ in
the two executions, and hence
$\alpha_H<1$ when evaluated at $t=t_f$ in Execution 4.
This gives the required contradiction (violating (1) from Definition
\ref{d:eaac}).


\section{Impossibility result for the partially synchronous setting: The proof of Theorem~\ref{neg2}} \label{neg2proof}

Consider the quasi-permissionless setting and the partial synchrony
model and suppose the blockchain protocol
$(\Sigma,\mathcal{O},\mathcal{C},\mathcal{S})$ is $\rho_l$-resilient
for liveness.  We consider three non-empty and pairwise disjoint sets
of players, $X$, $Y$ and $Z$, such that:
\begin{itemize} 
\item $\mathcal{P}=X \cup Y \cup Z$,  $|\mathcal{P}|=n$ (say);
\item $|X|=|Z|=\lfloor n \rho_l \rfloor$ and;
\item $|Y|=n-2\lfloor n \rho_l \rfloor$. 
\end{itemize} 

\vspace{0.2cm} 
\noindent \textbf{The intuition}. For the sake of simplicity, consider
a pure proof-of-stake protocol, meaning a protocol that does not make
use of external resources (although the formal proof below treats the
general case). Suppose all players begin with a single unit of each
form of stake and are active at all timeslots. The players in $X$ and
$Z$ are honest, while the players in $Y$ are Byzantine.

Because we are in the partial synchrony model, we may suppose that
messages disseminated by players in $X$ are received by players in
$X \cup Y$ at the next timeslot, but are not received by players
in~$Z$ until after GST. Symmetrically, we may suppose that messages
disseminated by players in $Z$ are received by players in $Z \cup Y$
at the next timeslot, but are not received by players in $X$ until
after GST. Suppose that the players in $Y$ initially act towards those
in $X$ as if GST$=0$ but the players in~$Z$ are Byzantine and are not
disseminating messages. Because the players in $Z$ own at most a
$\rho_l$ fraction of the stake and the protocol is $\rho_l$-resilient
for liveness, the players in $X$ must eventually confirm a sequence of
transactions, which may include transactions transferring all stake
from players in $Y$. If the players in $Y$ simultaneously act towards
those in $Z$ as if GST$=0$ but the players in $X$ are Byzantine and
are not disseminating messages, then the players in $Z$ must
eventually confirm a sequence of transactions, which may include
transactions transferring all stake from players in $Y$. When GST
arrives, this means that the honest players see a consistency
violation, but by this time the players in $Y$ have already
cashed out all of their stake.

\vspace{0.2cm}
\noindent \textbf{The formal proof}.  
Fix arbitrary $\Gamma$-liquid investment functions and an arbitrary
valuation function.
We consider three protocol
executions in which all parameters remain the same unless explicitly
stated otherwise.  The player set is always $X \cup Y \cup Z$ (as
described above), and
each player uses a single identifier (which we identify with the
player).  All players begin with a single unit of each form of stake
(if there exist any such) and are always active.  Let
$\mathtt{tr}_1,\mathtt{tr}_2$ denote distinct transactions that are
benign in the sense of Section~\ref{stake} and preserve the total
amount of resources controlled by the players in each of $X$, $Y$, and
$Z$ (e.g., a simple transfer between two players of~$X$).
For~$i=1,2$, let~$\mathtt{T}_i$ denote a set of transactions such
that, no matter how they ordered,
$\mathtt{S}_h(\mathtt{S}^*_h,\mathtt{tr}_i \ast \mathtt{T}_i,y)=0$ for
all stake allocation functions~$\mathtt{S}_h \in \mathcal{S}$ and
players~$y \in Y$. (The sets $\mathtt{T}_1$ and $\mathtt{T}_2$ exist
according to the assumptions in Section~\ref{stake}.)  For the sake of
simplicity, we suppose that $\Delta=1$, but the proof is easily
adapted to deal with larger values for $\Delta$.


\vspace{0.1cm} 
\noindent \textbf{Execution 1}. GST$=0$. Players in $X$ and $Y$ are
honest. Players in $Z$ are Byzantine and do not disseminate messages.
Players in $X$ have a single unit of each form of external resource at
each timeslot, while players in $Y$ and $Z$ do not own external
resources.  The environment sends a single transaction $\mathtt{tr}_1$
to a player $p_1\in X$ at timeslot 1.  At timeslot $2T_l$, the
environment sends the transactions in $\mathtt{T}_1$ to $p_1$.

\vspace{0.1cm} 
\noindent \textbf{Execution 2}.  This is similar to Execution 1, but
with the roles of $X$ and $Z$ reversed. GST$=0$. Players in $Y$ and
$Z$ and honest. Players in $X$ are Byzantine and do not disseminate
messages. 
Players in $Z$ have a single unit of each
form of external resource at each timeslot, while players in $X$ and
$Y$ do not own external resources.  The environment sends a single
transaction $\mathtt{tr}_2$ to a player $p_2\in Z$ at timeslot 1.  At
timeslot $2T_l$, the environment sends the transactions in
$\mathtt{T}_2$ to $p_2$.

\vspace{0.1cm} 
\noindent \textbf{Execution 3}. The execution is specified as follows: 
\begin{itemize} 

\item[-] Players in $X \cup Z$ are honest, while players in $Y$ are
  Byzantine.

\item[-] GST=$3T_l+\Gamma$. Any message disseminated by a player in $X$ at
  any timeslot $t$ is received by players in $X \cup Y$ at the next
  timeslot, but is not received by players in $Z$ until
  $\text{max} \{ \text{GST}, t+1 \}$. Any message disseminated by a
  player in $Z$ at any timeslot $t$ is received by players in
  $Y \cup Z$ at the next timeslot, but is not received by players in
  $X$ until $\text{max} \{ \text{GST}, t+1 \}$.

\item[-] 
Players in $X \cup Z$ have a single unit of each form of
  external resource at each timeslot, while players in $Y$ do not own
  external resources.

\item[-] The environment sends $\mathtt{tr}_1$ to $p_1$ and
  $\mathtt{tr}_2$ to $p_2$ at timeslot 1. At timeslot $2T_l$, the
  environment sends the transactions in $\mathtt{T}_1$ to $p_1$
and the transactions in $\mathtt{T}_2$ to
  $p_2$.
The environment then sends no further transactions.

\item[-] Each Byzantine player simulates two honest players. The first
  of these players acts honestly, except that they pretend messages
  sent by players in $Z$ prior to GST do not arrive until GST.  A
  message disseminated by this player at any timeslot $t$ is received
  by players in $X\cup Y$ at the next timeslot, but is not received by
  players in $Z$ until $\text{max} \{ \text{GST}, t+1 \}$. The second
  of these simulated players acts honestly, except that they pretend
  messages sent by players in $X$ prior to GST do not arrive until
  GST.  A message disseminated by this player at any timeslot $t$ is
  received by players in $Y\cup Z$ at the next timeslot, but is not
  received by players in $X$ until $\text{max} \{ \text{GST}, t+1 \}$.

\end{itemize} 

\vspace{0.2cm}
\noindent \textbf{Analysis}. Note that, in Execution 1, the adversary
is $\rho_l$-bounded. Since GST$=0$, $\mathtt{tr}_1$ must be confirmed
for all honest players by timeslot $1+T_l$ (while $\mathtt{tr}_2$ is
not, because this transaction is not received by any
player). Similarly, the transactions sent to $p_1$ at $2T_l$ must be
confirmed by $3T_l$; by the definition of~$\mathtt{T}_1$, 
players in $Y$ own no stake (of any form) at timeslots $\geq 3T_l$.
Because the investment functions corresponding to
$\mathtt{S}_1,\ldots,\mathtt{S}_j$ are assumed to be $\Gamma$-liquid,
$\mathtt{R}_h(p,3T_l+\Gamma) = 0$ for all such investment functions and
all $p \in Y$.

In Execution 2, the adversary is also $\rho_l$-bounded. Since GST$=0$,
$\mathtt{tr}_2$ must be confirmed for all honest players by timeslot
$1+T_l$ (while $\mathtt{tr}_1$ is not, because this transaction is not
received by any player). Similarly, the transactions sent to $p_2$ at
$2T_l$ must be confirmed by $3T_l$, meaning that players in $Y$ 
own no stake (of any form) at timeslots $\geq 3T_l$.
Because the investment functions corresponding to
$\mathtt{S}_1,\ldots,\mathtt{S}_j$ are assumed to be $\Gamma$-liquid,
$\mathtt{R}_h(p,3T_l+\Gamma) = 0$ for all such investment functions and
all $p \in Y$.

To complete the argument, note that, for players in $X$, Execution 3
is indistinguishable from Execution 1 at all timeslots $<$GST, i.e.\
those players receive precisely the same inputs, messages and oracle
responses at all timeslots $<$GST. It follows that $\mathtt{tr}_1$ is
confirmed for all players in $X$ by timeslot $1+T_l$, and that the
transactions of $\mathtt{T}_1$
are confirmed for
all players in $X$ by timeslot $3T_l$. Similarly, for players in $Z$,
Execution 3 is indistinguishable from Execution 2 at all timeslots
$<$GST.  It follows that $\mathtt{tr}_2$ is confirmed for all players
in $Z$ by timeslot $1+T_l$, and that the transactions of $\mathtt{T}_2$
are confirmed for all players in $Z$ by
timeslot~$3T_l$. 
It further follows that, in Execution~3, $\mathtt{R}_h(p,3T_l+\Gamma)
= 0$ for every investment function~$\mathtt{R}_h$ and $p \in
Y$.\footnote{The definition of a $\Gamma$-liquid investment
  function allows a
  player to cash out after {\em some} honest player sees a zero
  balance on-chain for at least $\Gamma$ consecutive time steps; 
  here, in fact, {\em all} honest players witness this.}

A consistency violation is first seen by honest players at GST. For all
choices of $t_f\geq$GST, the Byzantine players have
already cashed out their stakes (and never possessed any external
resources):
\[ 
\sum_{p \in Y} \sum_{i=1}^k \mathtt{R}_i(p,t_f) \cdot C_i=0
\]
and hence $\alpha_B=1$. The protocol therefore fails to be EAAC.

\vspace{0.2cm}
\noindent \textbf{Interpretation for the synchronous model}.  As
alluded to in Section~\ref{pos}, the proof of Theorem~\ref{neg2}
continues to hold in the synchronous model if $3T_l+\Gamma$ is less
than the worst-case message delay~$\Delta$. That is, to avoid the
impossibility result in Theorem~\ref{neg2}, either the time to
transaction confirmation or the time to recoup an investment off-chain
(following a transaction confirmation on-chain) must scale with the
worst-case message delay. For example, if typical network delays are
much smaller than worst-case delays and the speed of transaction
confirmation in some PoS protocol scales with the former, then the
``cooldown period'' required before stake can be liquidated must scale
with the latter (as it does, roughly, in the current Ethereum
protocol).

\section{A possibility result: the proof of Theorem~\ref{posres}}
\label{posproof}

\subsection{Review of (permissioned) Tendermint} \label{Treview} 

To motivate the definition of our PoS protocol, in this subsection we
give a somewhat informal description of a simplified specification
of (permissioned) Tendermint.  In this subsection, we therefore
consider a set of $n$ players, of which at most $f$ may be Byzantine,
where $n= 3f+1$.

\vspace{0.2cm} 
\noindent \textbf{Technical preliminaries}.  In Section \ref{setup},
we supposed that a player's state transition at a given timeslot can
depend on the number of the current timeslot. This modeling choice
serves to make our impossibility results as strong as possible, but is
a stronger assumption than is sometimes made when working in the
partially synchronous model.  For the sake of simplicity, we present a
PoS version of Tendermint which makes use of this assumption to avoid
the necessity of using a special procedure for \emph{view
  synchronization}.\footnote{If one wanted to drop this assumption of the
model, then one would need to augment the proof given here so as to
implement an appropriately modified version of some standard protocol
for view synchronization (e.g.\ \cite{naor2021cogsworth})). We avoid
such matters, as complicating the proof in this way would
distract from the primary considerations of the paper.} 
Similarly, we do not concern ourselves with issues of efficiency
(e.g.\ communication complexity and latency) and aim only to present
the simplest protocol meeting the required conditions.

In what follows, we assume that all messages are signed by the player
disseminating them.

\vspace{0.2cm} 
\noindent \textbf{Views in Tendermint}.   
The protocol instructions are divided into \emph{views}. Within each
view, we run two stages of voting, each of which is an opportunity for
players to vote on a \emph{block} (or blocks) of transactions proposed
by the \emph{leader} for the view (all these notions will be
formalized in Section \ref{formal}).  If the first stage of voting
establishes a \emph{stage 1 quorum certificate} (QC) for a block $b$,
then the second stage may establish a \emph{stage 2} QC for $b$.
Instructions within views are deterministic, so to ensure the protocol
as a whole is deterministic one can simply specify a system of
rotating leaders, e.g.\ if the players are $\{ p_0,\dots, p_{n-1} \}$,
then one can specify that the leader of view $v$, denoted
$\mathtt{lead}(v)$, is $p_i$, where $i= v \text{ mod }n$.

Each block of transactions $b$ records the view to which it
corresponds in the value $\mathtt{v}(b)$. If $Q$ is a QC for $b$, then
$\mathtt{v}(Q)=\mathtt{v}(b)$ and $\mathtt{b}(Q)=b$. The value
$\mathtt{s}(Q)$ records whether $Q$ is a stage 1 or a stage~2 QC.

Each player also maintains a variable $Q^{+}$ for the purpose of
implementing a \emph{locking} mechanism.  Upon seeing a stage 1 QC
while in view $v$, for $b$ which is proposed by the leader of view
$v$, an honest player $p$ sets $Q^{+}$ to be that stage 1 QC for $b$.

\vspace{0.2cm} 
\noindent \textbf{When view $v$ is executed}. View $v$ begins at time $3v\Delta$ and ends when view $v+1$ begins.

\vspace{0.2cm} 
\noindent \textbf{Block proposals in Tendermint}. Upon entering view $v$, the leader for the view lets $Q$ be the stage 1 QC amongst all those it has seen such that $\mathtt{v}(Q)$ is greatest. Let $b':= \mathtt{b}(Q)$. The leader then disseminates a block proposal $b$, which records its \emph{parent} $b'$ in the value $\mathtt{par}(b)$, which records the view corresponding to the block in the value $\mathtt{v}(b)$, and which records a QC for the parent block in the value $\mathtt{QCprev}(b):=Q$. Any honest player receiving a proposal $b$ while in view $v$ will regard it as \emph{valid} if it is signed by $\mathtt{lead}(v)$ and if $\mathtt{v}(b)=v$, $\mathtt{QCprev}(b)$ is a QC for the parent block $\mathtt{par}(b)$, and if $Q^+\leq \mathtt{v}(Q)$ for $Q=\mathtt{QCprev}(b)$.  If they regard the proposal as valid, honest players will then disseminate a stage 1 vote for $b$. Upon seeing a stage 1 QC for $b$ (a set of $2f+1$ stage 1 votes signed by distinct players) while in view $v$, they will set $Q^+$ to be that QC and will disseminate a stage 2 vote for $b$. Upon seeing any block with stage 1 and 2 QCs, honest players consider that block and all ancestors \emph{confirmed}. (Terms such as `ancestor' will be formally defined in Section \ref{formal}.) 

\vspace{0.2cm} 
\noindent \textbf{The genesis block}. We consider a fixed \emph{genesis block}, denoted $b_{g}$, with $\mathtt{v}(b_g):=0$. All players begin having already received a QC for the genesis block and with $Q^+$ equal to that QC for the genesis block. 

\vspace{0.2cm} 
An informal version of the instructions is shown below: 

\vspace{0.4cm}
\hrulefill

\vspace{0.3cm} 
\noindent \textbf{The instructions for player $p$ in view $v\geq 1$}.

\vspace{0.2cm}
\noindent \textbf{At timeslot} $t=3v\Delta$: If $p=\mathtt{lead}(v)$, then disseminate a new block, as specified above. 

\vspace{0.1cm}
\noindent \textbf{At timeslot} $t=3v\Delta+\Delta$: If $p$ has seen a first valid block $b$ for view $v$ signed by $\mathtt{lead}(v)$, i.e. if $\mathtt{v}(b)=v$, $\mathtt{QCprev}(b)$ is a QC for the parent block $\mathtt{par}(b)$, and if $Q^+\leq \mathtt{v}(Q)$ for $Q=\mathtt{QCprev}(b)$, then $p$ disseminates a stage 1 vote for $b$. 

\vspace{0.1cm}
\noindent \textbf{At timeslot} $t=3v\Delta+2\Delta$: If $p$ has seen a stage 1 QC for a block $b$ signed by $\mathtt{lead}(v)$, then set $Q^+$ to be that QC and disseminate a stage 2 vote for $b$.

\vspace{0.3cm} 
\hrulefill

\vspace{0.4cm} 
\noindent \textbf{Consistency and liveness for Tendermint}. To argue
that the protocol satisfies consistency, suppose towards a
contradiction that there exists a least $v$ such that:
\begin{itemize} 
\item Some $b$ with $\mathtt{v}(b)=v$ receives stage 1 and 2 QCs, $Q_1$ and $Q_2$ say. 
\item For some least $v'\geq v$, there exists $b'$ such that  $b'$ is  \emph{incompatible} with $b$ (i.e.\ neither $b$ or $b'$ are ancestors of each other), with $\mathtt{v}(b')=v'$ and $\mathtt{QCprev}(b')=Q_0$ (say), and the block $b'$ receives a stage 1 QC, $Q_3$ say. 
\end{itemize} 
If $v=v'$ then, since $Q_1$ and $Q_3$ both consist of $2f+1$ votes and
$n= 3f+1$, some honest player must have votes in both $Q_1$ and
$Q_3$. This gives a contradiction since each honest player sends at
most one stage 1 vote in each view. So, suppose $v'>v$. Then some
honest player $p$ must have votes in both $Q_2$ and $Q_3$.  This gives
the required contradiction, since $p$ must set $Q^+$ so that
$\mathtt{v}(Q^+)=v$ while in view $v$, but our choice of $(v,v')$
means that $\mathtt{v}(Q_0)< v$, so that $p$ would not regard the
proposal~$b'$ as valid while in view $v'$ and would not produce a vote
in $Q_3$.

To argue briefly that the protocol satisfies liveness, let $v$ be a
view with honest leader $p$, which begins at time at least $\Delta$
after GST. Amongst all the values $Q^+$ for honest players at time
$3v\Delta -\Delta$, let $Q$ be that such that $\mathtt{v}(Q)$ is
greatest. Then $p$ must see $Q$ by the beginning of view $v$, and will
therefore produce a block proposal $b$ that all honest players regard
as valid (since they cannot subsequently have updated their local
value $Q^+$). All honest players will therefore produce stage 1 and
stage 2 votes for $b$, meaning that $b$ is confirmed for all honest
players.

\subsection{Review of the Dolev-Strong protocol} \label{DSreview} 


In our protocol, honest players will reach consensus following a
consistency violation using a recovery procedure based on
(repeated instances of) a variant 
of the classic protocol of Dolev and Strong.  In a first instance of
the Dolev-Strong protocol, a designated \emph{leader} is asked to
propose an updated genesis block. If this instance results in
consensus amongst honest players on such a block~$b$,
then the honest majority of validating stake that remains after
slashing can ``sign off'' on this value, producing a form of
\emph{certificate} in support of~$b$.
If not, a second leader for another instance of the Dolev-Strong
protocol is then asked to propose an updated genesis block, and so on.

\vspace{0.2cm} In this section, we briefly review the Dolev-Strong
protocol. The version we present below is a general form of the
protocol for solving Byzantine Broadcast \cite{lamport2019byzantine}
given a set $P$ of $n$ players, of which at most $f$ may be Byzantine,
and one of which is the designated \emph{leader}.  In our recovery
procedure, a modified form of the protocol will be used, in which
honest players will ignore values proposed by leaders unless they are
a suitable proposal for an updated genesis block.

\vspace{0.2cm} 
\noindent \textbf{Notation for signed messages}.  For any message $m$, and for distinct players $p_1,\dots,p_t$, we let $m_{p_1,\dots,p_t}$ be $m$ signed by $p_1,\dots,p_t$ in order, i.e., for the empty sequence $\emptyset$, $m_{\emptyset}$ is $m$, and for each $i\in [1,t]$, $m_{p_1,\dots,p_i}$ is $m_{p_1,\dots,p_{i-1}}$ signed by $p_i$. Let $X_t$ be the set of all messages of the form $m_{p_1,\dots,p_t}$ such that  $p_1,\dots p_t$ are all distinct players in $P$ and $p_1$ is the leader.  

\vspace{0.2cm} 
\noindent \textbf{Initial setup}.  Each player  $p$ maintains a set $O_p$, which can be thought of as the set of values that $p$ recognises as having been sent by the leader, and which is initially empty.  
The leader begins with an input value $z$.

\vspace{0.2cm} 
\noindent \textbf{The instructions for player $p$}.

\vspace{0.1cm}

\noindent \textbf{Time $0$}. If $p$ is the leader, then disseminate  $z_p$ to all players and enumerate $z$ into $O_p$.

\noindent \textbf{Time $t\Delta^*$ with $1\leq t \leq f+1$}. Consider the set of messages $m\in X_t$   received by time $ t\Delta^*$. For each such message $m=y_{p_1,\dots ,p_t}$, if $y\notin O_p$,  proceed as follows: Enumerate $y$ into $O_p$ and, if $t<f+1$, disseminate $m_{p}$ to all players. 
 
\vspace{0.2cm}
\noindent \textbf{The output for player $p$}.  After executing all
other instructions at time $(f+1)\Delta^*$, $p$ outputs $y$ if~$O_p$
contains the single value $y$, and otherwise $p$ outputs $\bot$.

\vspace{0.2cm} 
\noindent \textbf{Proving that the protocol functions correctly}.  The
key claim is that all honest players produce the same output (possibly
$\bot$), whether the leader is honest or Byzantine. If the leader is
honest and begins with input $z$, then every honest player $p$ other
than the leader enumerates $z$ into $O_p$ at time~$\Delta^*$, and does
not enumerate any other value into this set at any time. All honest
players therefore produce a common output $z$ which, moreover, is the input of
the honest leader.
 
We can complete the proof by showing that, if any honest player $p$
enumerates a particular value~$y$ into $O_p$, then all honest players
do so (even if the leader is Byzantine). There are two cases to
consider:
\begin{itemize} 

\item \textbf{Case 1}. Suppose that some honest $p$ first enumerates
  $y$ into $O_p$ at time $t\Delta^*$ with $t<f+1$. In this case, $p$
  receives a message of the form $m=y_{p_1,\dots ,p_t}\in X_t$ at
  $t\Delta^*$. Player $p$ then adds their signature to form a message
  with $t+1$ distinct signatures and disseminates this message to all
  players. This means every honest player $p'$ will enumerate $y$ into
  $O_{p'}$ by time $(t+1) \Delta^*$.

\item \textbf{Case 2}. Suppose next that some honest $p$ first
  enumerates $y$ into $O_p$ at time $(f+1)\Delta^*$. In this case, $p$
  receives a message of the form $m=y_{p_1,\dots ,p_{f+1}}\in X_{f+1}$
  at time $(f+1)\Delta^*$, which has $f+1$ distinct signatures
  attached. At least one of those signatures must be from an honest
  player $p'$ (since there are at most $f$ Byzantine players), meaning
  that Case 1 applies w.r.t.\ $p'$.

\end{itemize} 

\subsection{Overview of further required changes} \label{furtherc} 

 Recall that PosT is our proof-of-stake version of Tendermint, used to
prove Theorem~\ref{posres}.
Section~\ref{six} described, at a high level, some of the innovations
involved: the holding of stake `in escrow', the use of epochs,
certificates of guilt, the use of a recovery procedure, and so on.
One further significant point concerns how to deal appropriately with
the change of validators at the end of each epoch.

\vspace{0.2cm}
\noindent \textbf{How to complete an epoch}. The complication is as follows. 

\vspace{0.1cm}
\noindent \emph{The problem}. We claimed in Section \ref{six} that,
should message delay prior to GST be bounded by $\Delta^*$, any
consistency violation during epoch $e$ will lead to active honest
players receiving a corresponding certificate of guilt within time
$2\Delta^*$ of any honest player beginning epoch $e+1$.  The basic
idea to specify the length of each epoch is therefore that we wish to
choose some $x$ such that $x\Delta>\Delta^*$, so that $3x\Delta$ (the
time to complete $x$ views\footnote{In fact, the time to complete each
  view will be $4x\Delta$ once we add in an extra stage of voting for
  each view.}) is certainly larger than $2\Delta^*$, and then have
epoch~$e$ be responsible for confirming blocks of heights in
$(ex,(e+1)x]$ (block height will be formally defined in Section
\ref{formal}).  The way the Tendermint protocol functions (at least as
specified in Section~\ref{Treview}), however, a block may only receive
a stage 1 QC and become confirmed because a descendant subsequently
receives a stage 2 QC. If $b$ is of height $(e+1)x$, and so
potentially ends an epoch, and if $b$ only receives a stage 1 QC, then
the question becomes, who do we have propose and vote on children of
$b$? We cannot allow the validator change at the end of the epoch to
be dictated by the transactions in $b$ and its ancestors (and
immediately have the new validators propose descendants of $b$),
because $b$ is not confirmed. Another block $b'$ of height $h$ might
also receive a stage 1 QC, leading to different opinions as to who
should be the validator set for the next epoch.

\vspace{0.2cm}
\noindent \emph{The solution}. The solution we employ is to allow the
validators for epoch $e$ to continue proposing and voting on blocks of
heights greater than $(e+1)x$, until they produce a confirmed block of
height $(e+1)x$. Once they do so, the blocks they have produced of
height $>(e+1)x$ (together with the votes for those blocks) are kept
as part of the `chain data' that verifies the validity of the chain,
but are not used to constitute part of the chain of confirmed
transactions, i.e. the validators for the next epoch begin building
above $b$ of height $(e+1)x$. It is crucial that the transactions in
these extra blocks not be considered confirmed because the validators
for the next epoch have not been present to implement the locking
mechanism during epoch $e$, meaning that if one of these blocks~$b'$
is confirmed, then we cannot guarantee that at least a third of the
new validators (weighted by stake) are honest and locked on
$b'$.\footnote{To see the issue, suppose $b'$ and $b''$ are both blocks of height $>(e+1)x$ which
are produced during epoch $e$, such that $b''$ is a child of $b'$, and
which both receive stage 2 QCs; this is possible because the stage 2
QC for $b''$ may be produced during asynchrony before GST and before
any player has seen the stage 2 QC for $b'$. If we were to regard the
transactions in $b'$ and $b''$ as confirmed, then some of the new
validators may initially see $b''$ as confirmed, while most have only
seen a stage 2 QC for $b'$. The new validators might then confirm new
blocks that are descendants of $b'$ incompatible with $b''$.}

\vspace{0.2cm}
\noindent \textbf{The need for three stages of voting in each
  view}. Similar to HotStuff \cite{yin2019hotstuff}, we will use three
stages of voting in each view. The motivation for doing so, however,
is different than for HotStuff (where the goals were
performance-driven, specifically to achieve a form of `optimistic
responsiveness'). The motivation here is to ensure that, if there is a
consistency violation amongst the blocks for epoch $e-1$, then active
honest players will receive a corresponding `certificate of guilt'
before the end of epoch $e$.\footnote{One concrete difference between
  the protocols is that,
  with a $\rho$-bounded adversary with $\rho < 1/3$, stage 2 QCs for
  conflicting blocks are possible in the HotStuff protocol but
  impossible in the PosT protocol.}

\vspace{0.1cm} 
\noindent \emph{The problem}. To see the problem if we only use two stages of voting, consider the following possible sequence of events. Suppose a block $b$ is proposed during view $v$ and receives a stage 1 QC. The block then receives stage 2 votes from a fraction of the honest stake that is small, but sufficient that the adversary can produce a stage 2 QC for $b$ at any later time of its choosing. 
In view $v+1$, a block $b'$ that is incompatible with $b$ then receives stage 1 and stage 2 QCs, ending the epoch. Note that, in this case,  the adversary can produce the stage 2 QC for $b$ and the corresponding consistency violation after any delay of its choosing. Once it does so, the stage 2 QC for $b$, together with the stage 1 QC for $b'$, will constitute a certificate of guilt (so long as votes are defined to include the $\mathtt{QCprev}$ value for the block they vote on), but by the time 
this is seen by honest players, the adversary's validating stake may no longer be in escrow. 

\vspace{0.1cm} 
\noindent \emph{The solution}. Now suppose that we use three stages of
voting within each view and that, as before, an honest player sets
their lock (redefines $Q^+$) upon seeing a stage 1 QC for the block
proposal during view $v$. Suppose message delay prior to GST is
bounded by $\Delta^*$. A block is confirmed when some descendant
(possibly the block itself) receives stage 1, 2 and 3 QCs.  A player
enters epoch $e+1$ upon seeing any confirmed block for epoch $e$ of
height $(e+1)x$.  Suppose there exists a consistency violation amongst
the epoch $e$ blocks. 
In this case, there must exist a least $v$ and a
least $v'\geq v$ such that some block $b$ with $\mathtt{v}(b)=v$
receives stage 1, 2 and 3 QCs, and some incompatible block $b'$ with
$\mathtt{v}(b')=v'$ also receives stage 1, 2 and 3 QCs. 
Crucially, because $\rho^* < 2/3$, the production of any QC requires
the participation of at least one honest player.
Let $v''$ be
the least view $\geq v$ such that some block $b''$ that is
incompatible with $b$ (with $\mathtt{v}(b'')=v''$) receives a stage 1
QC, $Q_1$ say, that is seen by some honest player $p_1$ before they
enter epoch $e+1$. Note that $v''\leq v'$. Let $Q_2$ be a stage 2 QC
for~$b$ that is seen by an honest player $p_2$ before sending a stage
3 vote during view~$v$. Since we will be able to argue (essentially
just as in Section \ref{Treview}) that $Q_1$ and $Q_2$ constitute an
appropriate certificate of guilt, it remains to show that these QCs
will be seen by all active honest players within time $2\Delta^*$ of
any honest player entering epoch $e+1$.


 If  some first honest player enters epoch $e+1$  at timeslot $t$, then all active honest players do so by  $t+\Delta^*$. This means $p_i$ (for $i\in \{ 1,2 \}$, as specified above) must enter epoch $e+1$  by $t+\Delta^*$ and, since $p_i$ sees $Q_i$ before $t+\Delta^*$, $Q_i$ must be seen by all active honest players by time $t+2\Delta^*$. All active honest players therefore see $Q_1$ and $Q_2$ by time  $t+ 2\Delta^*$, as required.  

\vspace{0.3cm} 
\noindent \textbf{Overall summary of the PosT protocol}. Before giving a formal protocol specification, we informally review the different phases of the PosT protocol. 

\vspace{0.1cm} 
\noindent \emph{Before the recovery procedure is triggered.} During
``normal'' operation (i.e., while not under attack), PosT acts
essentially like a version of
Tendermint, but with three phases of voting in each view, and with the
set of players that are responsible for proposing and voting on blocks
(i.e., the `validators') changing with each epoch. To act as a validator,
players must place stake in-escrow.  Removal of stake from escrow is
subject to a delay of one epoch, with the duration of each epoch
greater than~$2\Delta^*$. Since any consistency violation during epoch
$e$ will be seen by honest players before epoch $e+1$ is completed,
this suffices to ensure that Byzantine players contributing to a
consistency violation can be `slashed' before their stake is removed
from escrow. Since PosT acts essentially like Tendermint in all other
respects, in the case that the adversary is $\rho$-bounded for
$\rho<1/3$, the protocol achieves liveness with respect to a parameter
$T_l$ that depends on~$\Delta$ only (i.e., 
$T_l$ is independent of $\Delta^*$) and also satisfies consistency.

\vspace{0.1cm} 
\noindent \emph{Once the recovery procedure is triggered.} If a
consistency violation occurs in some least epoch $e$, then active
honest players will see a corresponding certificate of guilt before
epoch $e+1$ completes, and will trigger the recovery procedure. Once
the recovery procedure is triggered, the validators for epoch $e$ then
carry out repeated instances of the Dolev-Strong protocol to determine
an updated genesis block. This updated genesis block will contain all
transactions confirmed prior to epoch $e$, and also certificates of
guilt for at least 1/3 of the validating stake for epoch $e$, and will
be signed by the honest majority of validating stake that remains
after slashing (using that $\rho^*<2/3$). 
If desired, the main protocol can resume operation from this
new genesis block.

\subsection{The formal specification of the main protocol (prior to recovery)} \label{formal}
We give a specification that is aimed at simplicity of presentation, and do not concern ourselves with issues of efficiency, such as communication complexity and latency. It is
convenient to assume that whenever honest $p$ disseminates a message
at some timeslot $t$, $p$ regards that dissemination as having been
received (by $p$) at the next timeslot. We assume that all messages are signed by the player disseminating the message.

\vspace{0.2cm} 
\noindent \textbf{Transaction gossiping and the variable
  $\mathtt{T}^{\ast}$}. We assume that, whenever any honest player
first receives a transaction $\mathtt{tr}$, they disseminate
$\mathtt{tr}$ at the same timeslot. Each honest player $p$ maintains a
variable $\mathtt{T}^{\ast}$, which is the set of all transactions received by $p$ thus far.

\vspace{0.2cm} 
\noindent \textbf{The stake allocation function}.  We take as given an
initial distribution $\mathtt{S}^*$. The stake allocation function
$\mathtt{S}$ should be thought of as specifying the amount of stake a
player has in escrow and which has not been earmarked for removal from
escrow (after one epoch), i.e.\ a player's validating stake. For the
sake of simplicity, we suppose that the total number of units of
currency is fixed at some value $N\in \mathbb{N}_{>0} $ and that, for
some determined $x^*\in \mathbb{N}_{>0}$, each identifier can hold
either 0 or $N/x^*$ units of currency in escrow (this still allows a
player to put large amounts of stake in escrow by using multiple
identifiers).  Stake is added or removed from escrow via special
transactions: an \emph{add-to-escow} transaction corresponding to
identifier $id$ is used to put stake in escrow, while a
\emph{remove-from-escrow} transaction corresponding to $id$ is used to
remove stake from escrow. To specify how the stake may be updated, we
stipulate that players may create and disseminate `epoch' transactions
that will be used to mark the end of each epoch. `Certificates of
guilt' will also be formally defined in what follows as a special form
of transaction. For any given certificate of guilt~$\mathtt{G}$ and
any identifier $id$, $\mathtt{G}$ may or may not \emph{implicate}
the identifier $id$.  We require that $\mathtt{S}$ satisfies the following
conditions:\footnote{Throughout this paper, we use `$\ast$' to denote
  concatenation.}

\begin{enumerate} 

\item If $\mathtt{T}_2$ does not contain any epoch  transactions, then $\mathtt{S}(\mathtt{S}^*,\mathtt{T}_1\ast \mathtt{T}_2,id)= \mathtt{S}(\mathtt{S}^*,\mathtt{T}_1,id)$, i.e.\ \emph{$\mathtt{S}$ only updates at the end of an epoch}. 

\item Suppose $\mathtt{T}$ ends with an epoch transaction.
If $\mathtt{T}$ contains an add-to-escrow transaction corresponding to
$id$ that is not followed (in $\mathtt{T}$) by a remove-from-escrow  transaction corresponding to $id$, and also $\mathtt{T}$
  contains no certificate of guilt implicating $id$,
then $\mathtt{S}(\mathtt{S}^*,\mathtt{T},id)=N/x^*$.
Otherwise,
  $\mathtt{S}(\mathtt{S}^*,\mathtt{T},id)=0$.\footnote{While a
    remove-from-escrow transaction impacts a player's
    $\mathtt{S}$-balance immediately at the end of an epoch, the stake
    should be considered as remaining in escrow for one further epoch;
    this will be reflected in the definitions of the investment
    and valuation functions in Section \ref{instant}.}

\end{enumerate}

\vspace{0.2cm} 
\noindent \textbf{The length of an epoch}. Choose $x>2x^*$ (where $x^*$ is as specified above) and such that $x\Delta>\Delta^*$. We think of epoch $e$ as being responsible for confirming blocks of heights in $(ex,(e+1)x]$ (block height will be formally defined below). We will later specify a value $\mathtt{e}(b)$ corresponding to each block $b$, which records the epoch of the block. We will say a block $b$ is \emph{epoch $e$ ending} if $\mathtt{e}(b)=e$ and $b$ is of height $(e+1)x$. 


\vspace{0.2cm} 
\noindent \textbf{The genesis block}. We consider a fixed
\emph{genesis block} $b_g$.
We set
$\mathtt{h}(b_g)=\mathtt{e}(b_g)=\mathtt{v}(b_g)=0$. These values
indicate the \emph{height} of the block, and the epoch and view
numbers corresponding to the block respectively. We consider the empty
set, denoted $\emptyset$, a stage 1 QC for $b_g$, and set
$\mathtt{b}(\emptyset)=b_g$,
$\mathtt{e}(\emptyset)=\mathtt{v}(\emptyset)=0$. No set of messages is
a stage 2 QC for $b_g$. The block $b_g$ has no parent, has only itself
as an ancestor and is $M$-\emph{valid} for any set of messages $M$.

\vspace{0.2cm} 
\noindent \textbf{The variable} $M$. Each player $p$ maintains a
variable $M$, which is the set of all messages received by $p$ thus
far:  Intially $M$ is set to be $\{ b_g \}$.

\vspace{0.2cm} 
\noindent \textbf{Blocks}. A block $b$ other than the genesis block is entirely specified by the following values: 

\vspace{0.1cm} 
\noindent $\mathtt{h}(b)$: The height of $b$.

\vspace{0.1cm} 
\noindent $\mathtt{v}(b)$: The view corresponding to $b$. 

\vspace{0.1cm} 
\noindent $\mathtt{e}(b)$: The epoch corresponding to $b$.

\vspace{0.1cm} 
\noindent $\mathtt{par}(b)$: The parent of $b$. We also call $b$ a \emph{child} of $\mathtt{par}(b)$. The ancestors of $b$ are $b$ and all ancestors of $\mathtt{par}(b)$. The \emph{descendants} of any block $b$ are $b$ and all the descendants of its children.  Two blocks are \emph{incompatible} if neither is an ancestor of the other. For $b$ to be $M$-valid it must hold that:
\begin{itemize}
\item $b\in M$;
    \item $\mathtt{par}(b)$ is $M$-valid;
    \item The ancestors of $b$  include $b_g$;
    \item  $h(b)=h(\mathtt{par}(b))+1$;
    \item If $\mathtt{par}(b)$ is epoch
      $e=\mathtt{e}(\mathtt{par}(b))$ ending, then  $\mathtt{e}(b)$ is
      $e$ or $e+1$;\footnote{Thus, a block whose parent is epoch $e$ ending can still belong to the same epoch. This is due to the considerations outlined in Section \ref{furtherc}, i.e.\ further blocks may need to be built in the same epoch so as to confirm the block that is epoch $e$ ending. Once that block $b$ is confirmed, the first blocks of the next epoch will be children of $b$.} 
    \item If $\mathtt{par}(b)$ is not epoch $e=\mathtt{e}(\mathtt{par}(b))$ ending, then $\mathtt{e}(b)=e$.
\end{itemize} 

\vspace{0.1cm} 
\noindent $\mathtt{QCprev}(b)$: This value plays a similar role as in Section \ref{Treview}. For $b$ to be $M$-valid: 
\begin{itemize} 
\item $M$ must contain $\mathtt{QCprev}(b)$;
\item $\mathtt{QCprev}(b)$ must be a stage 1 QC for $\mathtt{par}(b)$. 
\end{itemize}

\vspace{0.1cm} 
\noindent $\mathtt{T}(b)$: The sequence of transactions in $b$. For
$b$ to be $M$-valid, we require that $\mathtt{T}(b)$ ends with an
epoch transaction iff $b$ is epoch $\mathtt{e}(b)$
ending.\footnote{This condition is important because epoch
  transactions inform $\mathtt{S}$  of the end of
  an epoch.}

\vspace{0.1cm} 
\noindent $\mathtt{Tr}(b)$: The sequence of transactions in all ancestors of $b$. For $b$ to be $M$-valid, we require $\mathtt{Tr}(b)= \mathtt{Tr}(\mathtt{par}(b)) \ast \mathtt{T}(b)$.

\vspace{0.1cm} 
\noindent $\mathtt{Tval}(b)$: The sequence of transactions used to
determine who should vote on $b$. We set 
$\mathtt{Tval}(b_g)=\emptyset$.
For $b$ to be $M$-valid, we require: 
\begin{itemize}
\item If $\mathtt{e}(b)=\mathtt{e}(\mathtt{par}(b))$ then $\mathtt{Tval}(b)=\mathtt{Tval}(\mathtt{par}(b))$;
\item If $\mathtt{e}(b)=\mathtt{e}(\mathtt{par}(b))+1$ then $\mathtt{Tval}(b)=\mathtt{Tr}(\mathtt{par}(b))$. 
\end{itemize} 

%

 
  
\vspace{0.2cm} 
\noindent \textbf{Votes}.  A vote $V$ is entirely specified by the following values: 

 \vspace{0.1cm} 
\noindent $\mathtt{b}(V)$: The block for which  $V$ is a vote. 

 \vspace{0.1cm} 
\noindent $\mathtt{c}(V)$: The amount of stake corresponding to the vote. 

 \vspace{0.1cm} 
\noindent $\mathtt{s}(V)$: The stage of the vote (1, 2 or 3). 

 \vspace{0.1cm} 
\noindent $\mathtt{id}(V)$: For the vote to be \emph{valid}, it must be signed by $\mathtt{id}(V)$.

 \vspace{0.1cm} 
\noindent $\mathtt{vprev}(V)$:  For the vote to be valid, this must equal $\mathtt{v}(\mathtt{QCprev}(\mathtt{b}(V)))$.

\vspace{0.2cm} 
\noindent \textbf{Quorum certificates}. Let $\mathtt{T}:= \mathtt{Tval}(b)$ and set: 
\[ \mathtt{N}(b):= \sum_{p\in \mathcal{P}} \sum_{id\in \mathtt{id}(p)}
  \mathtt{S}(\mathtt{S}^*,\mathtt{T},id). \]
For any block $b$ other than $b_g$, and for $s\in \{ 1, 2, 3 \}$, we
say a set of valid votes $Q$ is a stage $s$ QC for~$b$ if the
following conditions are all satisfied:

\begin{enumerate}

\item[(i)] For each $V\in Q$ we have $\mathtt{b}(V)=b$, and
  $\mathtt{s}(V)=s$;

\item[(ii)] For each
  $V\in Q$ it holds that $\mathtt{S}(\mathtt{S}^*,\mathtt{T},\mathtt{id}(V)) =\mathtt{c}(V)$;
  
  \item[(iii)] There do not exist $V\neq V'$ in $Q$ with $\mathtt{id}(V)=\mathtt{id}(V')$; 

\item[(iv)]   $\sum_{V\in Q} \mathtt{c}(V)\geq  \frac{2}{3}\mathtt{N}(b)$.


\end{enumerate} 

\noindent If $Q$ is a stage
$s$ QC for $b$, we define $\mathtt{b}(Q):=b$, $\mathtt{e}(Q):=\mathtt{e}(b)$, $\mathtt{v}(Q):=\mathtt{v}(b)$,  $\mathtt{vprev}(Q)=\mathtt{v}(\mathtt{QCprev}(b))$,  and
  $\mathtt{s}(Q):=s$.

 \vspace{0.1cm} 
 \noindent In an abuse of notation, we'll say $Q\in M$ when every element of $Q$ is in $M$. 

 \vspace{0.1cm} 
 \noindent \emph{Recall that $\mathtt{Tval}(b)$ is the sequence of transactions used to determine who should vote on $b$. So, this sequence determines the `validating stake'. The conditions above stipulate that votes from 2/3 of the validating stake are required to form a QC.}

\vspace{0.2cm} 
\noindent \textbf{Defining the confirmation rule}. A block $b$ is $M$-confirmed if there exists a descendant $b'$ of~$b$ with $\mathtt{e}(b')=\mathtt{e}(b)$,  such that $b'$ is $M$-valid with stage 1, 2 and 3 QCs in $M$, and if it also holds that $b$ is of height $\leq (\mathtt{e}(b)+1)x$. Let $b$ be of greatest height amongst the $M$-confirmed blocks. The sequence of $M$-confirmed transactions is $\mathtt{Tr}(b)$.\footnote{Once an honest player $p$ sees that the recovery procedure has been triggered and has terminated, outputting a new genesis block $b_g'$, it is technically convenient for $p$ to regard $b_g'$ as the only confirmed block. If a consistency violation occurs in some least epoch $e$, then $b_g'$ will contain all transactions confirmed by the end of epoch $e-1$.}  

For any block $b$ other than $b_g$ to be $M$-valid we require that: 
\begin{itemize}
\item If $\mathtt{e}(b)>\mathtt{e}(\mathtt{par}(b))$, then $\mathtt{par}(b)$ is $M$-confirmed. 
\end{itemize}

\vspace{0.2cm} 
\noindent \textbf{$M$-validity}. A block is $M$-valid if it satisfies all of the conditions for $M$-validity listed above (in this section). 

\vspace{0.2cm} 
\noindent \textbf{Defining certificates of guilt}. A certificate of guilt for epoch $e$ (considered a special form of transaction) is a pair of QCs $(Q,Q')$ such that: 
\begin{itemize} 
\item  $Q$ is a stage 2 QC for some $b$ with $\mathtt{e}(b)=e$ and $\mathtt{v}(b)=v$ (say);  
\item  $Q'$ is a stage 1 QC for some $b'$ with $\mathtt{e}(b')=e$ and $\mathtt{v}(b')=v'$ (say);
\item $b'$ is incompatible with $b$;
\item $v'\geq v$ and $\mathtt{vprev}(Q')<v$.  
\end{itemize} 
\noindent A certificate of guilt $(Q,Q')$ \emph{implicates} $id$ if there exist votes $V\in Q$, $V'\in Q'$ with $\mathtt{id}(V)=\mathtt{id}(V')=id$. 

 \vspace{0.2cm} 
\noindent \textbf{The local variable $Q^+$}. This plays a similar role as in Section \ref{Treview}. Initially $Q^+:=\emptyset$, i.e.\ $Q^+$ is set to be a stage 1 QC for $b_g$.

\vspace{0.2cm} 
\noindent \textbf{The local variable $\mathtt{e}$}. The value $\mathtt{e}$  records the current epoch  for $p$. Initially, $\mathtt{e}=0$. 

\vspace{0.2cm} 
\noindent \textbf{The local value $\mathtt{rec}$}. This is initially 0, and is set to 1 to indicate that $p$
should start executing instructions for the recovery procedure.

\vspace{0.2cm} 
\noindent \textbf{The $\mathtt{eupdate}$ procedure}.  The following procedure called $\mathtt{eupdate}$ is responsible for updating $\mathtt{e}$: 
\begin{enumerate} 
\item If $\mathtt{rec}=0$, then let $e$ be greatest such that there exists a block $b$ which is $M$-confirmed and epoch $e$ ending, or, if there exists no such $e$, let $e=-1$. Set $\mathtt{e}:=e+1$. If this is the first timeslot at which $\mathtt{e}=e+1$, then we say \emph{epoch $e+1$ begins at $t$ for $p$}. 
  
\item If $\mathtt{rec}=1$ and $M$ contains a certificate of guilt for epoch $\mathtt{e}-1$, set $\mathtt{e}:=\mathtt{e}-1$. 

\end{enumerate}

\vspace{0.1cm} 
\noindent  \emph{(1) above specifies the present epoch for $p$ before the recovery procedure is triggered. This is defined in the obvious way: $p$ looks to see which epoch it has seen completed and begins work on the next epoch. (2) updates the epoch when the recovery procedure is triggered because of a consistency violation in epoch $\mathtt{e}-1$. If the recovery procedure is triggered because of a consistency violation in epoch $\mathtt{e}$, there is no need to update $\mathtt{e}$.}

\vspace{0.2cm} 
\noindent \textbf{When $p$ is ready to enter the recovery procedure}. We say $p$ is ready to enter the recovery procedure if either: 
\begin{itemize}
\item   $M$ contains a certificate of guilt for epoch $\mathtt{e}-1$, or; 
\item Epoch $\mathtt{e}$ began for $p$ at a timeslot $t'<t-2\Delta^*$, and $M$ contains a certificate of guilt for epoch~$\mathtt{e}$.
\end{itemize} 

\vspace{0.1cm} 
\noindent \emph{The conditions above specify when the recovery procedure is triggered,  and stipulate that one must wait time $2\Delta^*$ after entering epoch $e$ before triggering the recovery procedure because of a consistency violation in epoch $e$.  This is to give sufficient time for a consistency violation in the previous epoch to be revealed, and to ensure that the recovery procedure is only triggered because of a consistency violation in a single epoch.}


\vspace{0.2cm} 
\noindent \textbf{The function $\mathtt{Qset}$}. This function is used  to update $Q^+$ after entering a new epoch, and is defined as follows. Amongst the $M$-confirmed and epoch $\mathtt{e}-1$ ending blocks, choose $b$ such that $\mathtt{v}(b)$ takes the greatest value.  Let $Q$ be some stage 1 QC for $b$ and define $\mathtt{Qset}(M,\mathtt{e}):=Q$.

\vspace{0.2cm} 
\noindent \textbf{Defining} $\mathtt{lead}(M,\mathtt{e},v)$. Player $p$'s belief as to who should be the leader for view $v$ depends on their values $M$ and $\mathtt{e}$. If there does not exist a unique $b\in M$ which is $M$-confirmed and epoch $\mathtt{e}-1$ ending, then $\mathtt{lead}(M,\mathtt{e},v)$ is undefined.\footnote{Generally, we write $x\uparrow$ to indicate that the variable $x$ is undefined, and we write $x\downarrow$ to indicate that $x$ is defined.} Otherwise, let $b$ be such a block and set $\mathtt{T}:=\mathtt{Tr}(b)$ and $v:=\mathtt{v}(b)$. Recall our assumption that, for some determined $x^*\in \mathbb{N}_{>0}$, if  $\mathtt{S}(\mathtt{S}^*,\mathtt{T},id)>0$ then $\mathtt{S}(\mathtt{S}^*,\mathtt{T},id)=N/x^*$. Let $id_0,\dots,id_{k-1}$ be an enumeration of the identifiers $id$ such that  $\mathtt{S}(\mathtt{S}^*,\mathtt{T},id)>0$. For $i\in \mathbb{N}_{>0}$, we define $\mathtt{lead}(M,\mathtt{e},v+i)$ to be the identifier $id_j$, where $j= i \text{ mod } k$.

\vspace{0.2cm} 
\noindent \textbf{Admissible blocks}. When $p$ determines whether to vote on a block $b$ in view $v$, it will only do so if the block is \emph{admissible}. At any point,  $p$ regards $b$ as admissible for view $v$ if all of the following conditions are satisfied: 
\begin{itemize} 
\item $b$ is $M$-valid; 
\item $\mathtt{v}(b)=v$ and $b$ is signed by $\mathtt{lead}(M,\mathtt{e},v)$; 
\item $\mathtt{v}(Q^+)\leq \mathtt{v}(\mathtt{QCprev}(b))$;
\item $\mathtt{e}(b)=\mathtt{e}$;
\end{itemize}

\vspace{0.2cm} 
\noindent \textbf{The $\mathtt{disseminate\  new \ block}$ procedure}. If $\mathtt{lead}(M,\mathtt{e},v)$ (as locally defined for $p$) is an identifier $id\in \mathtt{id}(p)$ then, at the beginning of view $v$, $p$ proceeds as follows in order to specify and disseminate a new block: 
\begin{enumerate} 
\item Amongst the $M$-valid $b'$ with a stage 1 QC in $M$ such that $\mathtt{e}(b')=\mathtt{e}$, choose $b'$ such that $\mathtt{v}(b')$ is maximal and set $Q$ to be a stage 1 QC for $b'$ in $M$. Or, if there exists no such block in $M$, set $Q:=\mathtt{Qset}(M,\mathtt{e})$ and let $b':=\mathtt{b}(Q)$. 
\item Let $\mathtt{T}$ be the set of transactions in
  $\mathtt{T}^*-\mathtt{Tr}(b')$. Let $b$ be a new block with
  $\mathtt{e}(b)=\mathtt{e}$, $\mathtt{v}(b)=v$, $\mathtt{par}(b)=b'$,
  $\mathtt{QCprev}(b)=Q$, $\mathtt{T}(b)$ the transactions of
  $\mathtt{T}$ (in some arbitrary order and, if appropriate, concluding
  with a transaction indicating the end of epoch~$\mathtt{e}$),
 and which satisfies all conditions for $M$-validity
  (note that these conditions suffice to specify $\mathtt{Tval}(b)$ and $\mathtt{h}(b)$). 

\item Disseminate $b$, signed by $id$. 
\end{enumerate}

\vspace{0.2cm} 
\noindent \textbf{The pseudocode}. The pseudocode is described in Algorithm 1.

\subsection{The formal specification of the recovery procedure} \label{rp} 

Upon setting $\mathtt{rec}:=1$, indicating that the recovery procedure
is triggered, player $p$ will run the $\mathtt{eupdate}$ procedure,
which will set $\mathtt{e}$ to be the least epoch in which a
consistency violation occurs. The aim of the recovery procedure is
then to reach consensus on an updated genesis block $b_g'$ 
at which the slashing instructions have been carried out.
To specify $b_g'$, one need only specify
the transactions in $b_g'$, although honest players will also produce
votes for $b_g'$, ensuring that it is uniquely determined.

\vspace{0.2cm} 
\noindent \textbf{$M$-admissible genesis proposals}.  Let $M$ and $\mathtt{e}$ be as locally defined for $p$ during the recovery procedure.  Let $b$ be an epoch $\mathtt{e}-1$ ending block that is $M$-confirmed, and set $\mathtt{T}:=\mathtt{Tr}(b)$. A genesis proposal $(b_g',i)$ is $M$-admissible if $\mathtt{T}(b_g')=\mathtt{T}\ast \mathtt{G}\ast \mathtt{tr}$, where $\mathtt{G}$ is a  certificate of guilt for epoch $\mathtt{e}$ and $\mathtt{tr}$ is an epoch transaction. The second coordinate $i$ in the genesis proposal $(b_g',i)$ indicates that the proposal is made during the $i^{\text{th}}$ instance of the Dolev-Strong protocol. 

\vspace{0.2cm} 
\noindent \textbf{Defining} $\mathtt{reclead}(M,i)$ \textbf{and} $\mathtt{signed}(M,i,j)$. We use a similar notation for signed messages as in Section \ref{DSreview}.  For any message $m$, and for distinct identifiers $id_1,\dots,id_{j}$, we let $m_{id_1,\dots,id_j}$ be $m$ signed by $id_1,\dots,id_j$ in order.
As above, let  $M$ and $\mathtt{e}$ be as locally defined for $p$ during the recovery procedure.  Let $b$ be an epoch $\mathtt{e}-1$ ending block that is $M$-confirmed, and set $\mathtt{T}:=\mathtt{Tr}(b)$. Let  $id_0,\dots,id_{k-1}$ be an enumeration of the identifiers $id$ such that  $\mathtt{S}(\mathtt{S}^*,\mathtt{T},id)>0$ and set $\mathtt{Id}(M):= \{ id_0,\dots,id_{k-1} \}$. For $i\in \mathbb{N}_{\geq 0}$, we define $\mathtt{reclead}(M,i)$ to be the identifier $id_j$, where $j= i \text{ mod } k$. This value specifies the leader for the $i^{\text{th}}$ instance of the Dolev-Strong protocol.

For  $i\in \mathbb{N}_{\geq 0}$ and $j\in \mathbb{N}_{\geq 1}$, we define $\mathtt{signed}(M,i,j)$ to be the set of all messages in $M$ which are of the form $m_{id_1',\dots, id_j'}$, such that $m$ is an $M$-admissible genesis proposal $(b_g',i)$ and $id_1',\dots,id_j'$ are all distinct members of $\mathtt{Id}(M)$, with $id_1'=\mathtt{reclead}(M,i)$.

\vspace{0.2cm} 
\noindent \textbf{The variable $O_p$}. This plays the same role as in Section \ref{DSreview}.

\vspace{0.2cm} 
\noindent \textbf{Votes for the final output}. Honest players will
disseminate \emph{output votes} for a value $b_g'$ that may be used as
an updated genesis block.  These output
votes $V$ are entirely specified by three values $\mathtt{b}(V)$,
$\mathtt{c}(V)$ and $\mathtt{id}(V)$, and must be signed by
$\mathtt{id}(V)$ to be valid.

\vspace{0.2cm} 
\noindent \textbf{The pseudocode}. The pseudocode is described in Algorithm 2.

\begin{algorithm} 
\caption{The instructions for player $p$ before the recovery procedure is triggered.}
\begin{algorithmic}[1]

            \State \textbf{At timeslot $1$:}

                 \State \hspace{0.3cm} Set $\mathtt{end}:=0$,  $\mathtt{rec}:=0$;   

    \State

            \State \textbf{At every timeslot $t$ if $\mathtt{end}=0$:}

         \State \hspace{0.3cm}  $\mathtt{eupdate}$;       \Comment{Defined in Section \ref{formal}, responsible for updating $\mathtt{e}$}

               \State \hspace{0.3cm} \textbf{If} epoch $\mathtt{e}$ begins at $t$ for $p$ \textbf{then} set $Q^+:=\mathtt{Qset}(M,\mathtt{e})$;    \Comment{Set $\mathtt{Q}^+$ at start of epoch}

         \State \hspace{0.3cm} \textbf{If} $p$ is ready to enter the recovery procedure \textbf{then}:  \Comment Condition defined in Section \ref{formal}
         
         \State \hspace{0.6cm}  Set $\mathtt{rec}:=1$ and perform $\mathtt{eupdate}$;            \Comment{Start recovery procedure}
            \State \hspace{0.6cm} Set $\mathtt{end}:=1$.

            \State

            \State \textbf{At timeslot $4v\Delta$ for $v\in \mathbb{N}_{>0}$ if $\mathtt{end}=0$}: 
            
             \State \hspace{0.3cm} \textbf{If} $\mathtt{lead}(M,\mathtt{e},v)\in \mathtt{id}(p)$ \textbf{then} $\mathtt{disseminate\ new\ block}$;     \Comment Disseminate a new block

          \State 

            \State \textbf{At timeslot $4v\Delta + \Delta$ for $v\in \mathbb{N}_{>0}$ if $\mathtt{end}=0$}:

             \State \hspace{0.3cm} Set $\mathtt{b}^*$ to be undefined; 

            \State \hspace{0.3cm} \textbf{If} there exists a first $b$ enumerated into $M$ which is admissible for view $v$ \textbf{then}:

            \State \hspace{0.6cm} Set $\mathtt{b}^*:=b$ and $\mathtt{T}:=\mathtt{Tval}(b)$, $Q:=\mathtt{QCprev}(b)$;

              \State  \hspace{0.6cm}  For each $id\in \mathtt{id}(p)$ such that $\mathtt{S}(\mathtt{S}^*,\mathtt{T},id)>0$: \Comment{Disseminate stage 1 votes}
        
       \State  \hspace{0.9cm}         Let $c:= \mathtt{S}(\mathtt{S}^*,\mathtt{T},id)$;  
        
              \State  \hspace{0.9cm}   Disseminate $V$ with $\mathtt{b}(V)=\mathtt{b}^*$, $\mathtt{c}(V)=c$, $\mathtt{s}(V)=1$, $\mathtt{id}(V)=id$, $\mathtt{vprev}(V)=\mathtt{v}(Q)$;

       \State 

            \State \textbf{At timeslot $4v\Delta + 2\Delta$ for $v\in \mathbb{N}_{>0}$ if $\mathtt{end}=0$}:

             \State \hspace{0.3cm}  \textbf{If} $\mathtt{b}^*\downarrow $ and $M$ contains a stage 1 QC for $\mathtt{b}^*$  \textbf{then}:

                \State \hspace{0.6cm} Set $Q^+$ to be a stage 1 QC for $\mathtt{b}^*$ in $M$; \Comment Set new lock

              \State  \hspace{0.6cm}  For each $id\in \mathtt{id}(p)$ such that $\mathtt{S}(\mathtt{S}^*,\mathtt{T},id)>0$: \Comment{Disseminate stage 2 votes}
        
       \State  \hspace{0.9cm}         Let $c:= \mathtt{S}(\mathtt{S}^*,\mathtt{T},id)$;  
        
              \State  \hspace{0.9cm}   Disseminate $V$ with $\mathtt{b}(V)=\mathtt{b}^*$, $\mathtt{c}(V)=c$, $\mathtt{s}(V)=2$, $\mathtt{id}(V)=id$, $\mathtt{vprev}(V)=\mathtt{v}(Q)$;   
            
       \State 

            \State \textbf{At timeslot $4v\Delta + 3\Delta$ for $v\in \mathbb{N}_{>0}$ if $\mathtt{end}=0$}:

             \State \hspace{0.3cm}  \textbf{If} $\mathtt{b}^*\downarrow $ and $M$ contains a stage 2 QC for $\mathtt{b}^*$  \textbf{then}:

              \State  \hspace{0.6cm}  For each $id\in \mathtt{id}(p)$ such that $\mathtt{S}(\mathtt{S}^*,\mathtt{T},id)>0$: \Comment{Disseminate stage 3 votes}
        
       \State  \hspace{0.9cm}         Let $c:= \mathtt{S}(\mathtt{S}^*,\mathtt{T},id)$;  
        
              \State  \hspace{0.9cm}   Disseminate $V$ with $\mathtt{b}(V)=\mathtt{b}^*$, $\mathtt{c}(V)=c$, $\mathtt{s}(V)=3$, $\mathtt{id}(V)=id$, $\mathtt{vprev}(V)=\mathtt{v}(Q)$;

\end{algorithmic}
\end{algorithm}

\begin{algorithm} 
\caption{Recovery procedure instructions for $p$: to be carried out when $\mathtt{rec}=1$.}
\begin{algorithmic}[1]

   \State \textbf{At timeslot $4\Delta((\mathtt{e}+2)x+1)$}:     \Comment Initialization
   
   \State \hspace{0.3cm}  Set $\mathtt{recend}:=0$;

   \State \hspace{0.3cm} Let $b$ be epoch $\mathtt{e}-1$ ending and $M$-confirmed; Set $\mathtt{T}:=\mathtt{Tr}(b)$;        
   
      \State  \hspace{0.3cm} Set $k=|\{ id:\ \mathtt{S}(\mathtt{S}^*,\mathtt{T},id)>0 \}|$; \Comment Number of participants
      
         \State  \hspace{0.3cm}   Set $t_i:=4\Delta((\mathtt{e}+2)x+1) +(k+1)i\Delta^*$, $i\in \mathbb{N}_{\geq 0}$;   \Comment First instance of Dolev-Strong begins at $t_0$
   
   \State

  \State \textbf{At timeslot $t_i$ for $i\in \mathbb{N}_{\geq 0}$ if $\mathtt{recend}=0$}:     \Comment  Begin $i^{\text{th}}$ instance of Dolev-Strong
  
  \State \hspace{0.3cm} Set $O_{p}=\emptyset$; 
  
  \State \hspace{0.3cm} \textbf{If} $\mathtt{reclead}(M,i)\in \mathtt{id}(p)$ \textbf{then}:  \Comment Propose updated genesis block
  
  \State \hspace{0.6cm}  Disseminate a genesis proposal $(b_g',i)$ signed by $\mathtt{reclead}(M,i)$ which is $M$-admissible;   
  
\State 

  \State \textbf{At timeslot $t_i +j\Delta^*$ for $i\in \mathbb{N}_{\geq 0}$ and $j\in [1,k]$  if $\mathtt{recend}=0$}:   
  
\State \hspace{0.3cm} For each message $m=y_{id_1',\dots ,id_j'}$ in $\mathtt{signed}(M,i,j)$:  \Comment  Multi-signed $M$-admissible proposals

\State \hspace{0.6cm}  \textbf{If} $y\notin O_p$ \textbf{and} $j<k$: 

\State \hspace{0.9cm} For each  $id\in \mathtt{id}(p)$ such that $\mathtt{S}(\mathtt{S}^*,\mathtt{T},id)>0$,  disseminate $m_{id}$;

\State \hspace{0.6cm}  \textbf{If} $y\notin O_p$, enumerate $y$ into $O_p$;

    \State 
    
      \State \textbf{At timeslot $t_i +k\Delta^*$ if $\mathtt{recend}=0$}:    \Comment End of $i^{\text{th}}$ instance of Dolev-Strong

    \State  \hspace{0.3cm}   \textbf{If} $O_p$ contains a single value $(b_g',i)$ \textbf{then}:      \Comment Send votes for updated genesis block

              \State  \hspace{0.6cm}  For each $id\in \mathtt{id}(p)$ such that $\mathtt{S}(S^*,\mathtt{T},id)>0$: 
        
       \State  \hspace{0.9cm}         Let $c:= \mathtt{S}(\mathtt{S}^*,\mathtt{T},id)$;  
        
              \State  \hspace{0.9cm}   Disseminate output vote $V$ with $\mathtt{b}(V)=b_g'$, $\mathtt{c}(V)=c$,  $\mathtt{id}(V)=id$;

\State  \hspace{0.6cm}  Set $\mathtt{recend}:=1$.         \Comment{Terminate}

\end{algorithmic}
\end{algorithm}

\subsection{Certificates of guilt and consistency} \label{posproof1}

In this section, we show that the protocol is consistent for $\rho$-bounded adversaries when $\rho<1/3$.  In the next section, we show liveness. 

\vspace{0.2cm}
\noindent \textbf{Further terminology}. We let $\mathcal{M}^*$ be the set of all messages disseminated during an execution. If $id\in \mathtt{id}(p)$ and $p$ is Byzantine, then we say $id$ is Byzantine. 

\begin{lemma} \label{lem1}
If a certificate of guilt in $\mathcal{M}^*$ implicates $id$, then $id$ is Byzantine. 
\end{lemma} 
\begin{proof}
Let $(Q,Q')$ be a certificate of guilt for epoch $e$, and let $b,b',v,v'$ be as specified in the definition of certificates of guilt in Section \ref{formal}. Suppose  there exist votes $V\in Q$ and $V'\in Q'$ with $\mathtt{id}(V)=\mathtt{id}(V')=id$ and let $p$ be such that $id\in \mathtt{id}(p)$. If $v=v'$, then the claim is immediate because honest players do not vote for two different blocks in the same view. If $v'>v$, then the claim follows because, if $p$ is honest, then it must set $Q^+:=v$ upon disseminating $V$ while in view~$v$, and then could not disseminate $V'\in Q'$ because $\mathtt{vprev}(Q')<v$, i.e. $p$ would not regard $\mathtt{b}(Q')$ as admissible when voting in view $v'$.  
\end{proof}

Next, recall that $\rho^*<2/3$. In particular, every quorum certificate must
contain at least one vote from an honest player. Thus, for $s=2,3$, no
stage $s$ QC can be formed for a block~$b$ without some honest player
previously seeing a stage $s-1$ QC for~$b$. Similarly, no stage 1 QC
can be formed for a block~$b$ without some honest player seeing a
stage 1 QC for $b$'s parent.

\begin{lemma} \label{lem2} 
Suppose two incompatible blocks for epoch $e$  are both $\mathcal{M}^*$-confirmed. Then $\mathcal{M}^*$ contains a certificate of guilt for epoch $e$. Let $t$ be the first timeslot at which any honest player enters epoch $e+1$ or sets their local value $\mathtt{rec}:=1$ (i.e. starts the recovery procedure).  If message delay prior to GST is bounded by $\Delta^*$, then all active honest players receive a certificate of guilt for epoch $e$ by time $t+2\Delta^*$. 
\end{lemma}
\begin{proof} 
 In the argument that follows, we restrict attention to blocks $b$ with $\mathtt{e}(b)=e$.  Given the conditions in the statement of the lemma,  there must exist a least $v$, a least $v'\geq v$, and incompatible blocks $b$ and $b'$ which both receive stage 1, 2 and 3 QCs in $\mathcal{M}^*$,  such that $\mathtt{v}(b)=v$ and $\mathtt{v}(b')=v'$.
 Let $v''$ be the least view $\geq v$ such that some block $b''$ that
 is incompatible with $b$ (with $\mathtt{v}(b'')=v''$) receives a
 stage 1 QC, $Q_1$ say,  that is seen by some honest player $p_1$
 before they enter epoch $e+1$ or start the recovery procedure. 
As~$b'$ is one such block, $v''\leq v'$. Let $Q_2$ be a stage 2 QC for $b$ that is seen by an honest player $p_2$ before disseminating a stage 3 vote during view $v$. Then $(Q_2,Q_1)$ is a certificate of guilt for epoch $e$ since, by our choice of $v''$, $\mathtt{vprev}(Q_1)<v$. 

It remains to show that $Q_1$ and $Q_2$ will be seen by all active honest players by time $t+2\Delta^*$.
 If  some first honest player enters epoch $e+1$ or starts the
 recovery procedure at timeslot $t$, then all active honest players do
 so by  $t+\Delta^*$. This means $p_i$ (for $i\in \{ 1,2 \}$, as
 specified above) must enter epoch $e+1$ or start the recovery
 procedure by $t+\Delta^*$.
Player $p_i$ must therefore have seen $Q_i$ before $t+\Delta^*$, and so $Q_i$ must be seen by all active honest players by time $t+2\Delta^*$. All active honest players therefore see $Q_1$ and $Q_2$ by time  $t+ 2\Delta^*$, as required.  
\end{proof} 

\vspace{0.2cm} 
\noindent \textbf{The validating stake}. Suppose there exists a unique $\mathcal{M}^*$-confirmed block $b$ that is epoch $e$ ending. By the \emph{validating stake} for epoch $e+1$, we mean the set of identifiers $id$ such that $\mathtt{S}(\mathtt{S}^*,\mathtt{Tr}(b),id)>0$. 

\begin{lemma} \label{lem3} 
PosT is consistent for $\rho$-bounded adversaries when $\rho<1/3$.
\end{lemma} 
\begin{proof} 
This follows almost directly from Lemmas \ref{lem1} and
\ref{lem2}. Suppose the adversary is $\rho$-bounded for $\rho<1/3$ and
that there exists a least epoch $e$ that sees a consistency violation,
i.e. such that incompatible blocks $b$ and $b'$ with
$\mathtt{e}(b)=\mathtt{e}(b')=e$  are both $\mathcal{M}^*$-confirmed.
Because QCs for blocks in epoch $e$ require votes from at least 2/3 of
the validating stake for epoch $e$, any certificate of guilt
implicates at least 1/3 of the validating stake for epoch $e$. By
Lemma~\ref{lem2}, $\mathcal{M}^*$ contains a certificate of guilt for
epoch~$e$, which implicates at least 1/3 of the validating stake for
epoch $e$. By Lemma 1, all the implicated identifiers must be
Byzantine.   This gives the required contradiction.
\end{proof}

\subsection{Establishing liveness} \label{posproof2}

Suppose the adversary is $\rho$-bounded for $\rho<1/3$. Lemma
\ref{lem3} establishes consistency in this case, which means that, if
two honest players are in the same epoch $e$, then their local values
$\mathtt{lead}(M,e,v)$ agree for all $v$.  We may therefore refer
unambiguously to the value $\mathtt{lead}(e,v)$ for each $e,v$. Let
$k^*$ take the minimum value $\geq 1$ such that, for any $e$ and $v$,
at least one $v' \in (v,v+k^*]$ must satisfy the condition that
$\mathtt{lead}(e,v')$ is honest. Note that $k^*$ is a function of
$x^*$ (the maximum number of validators), but is independent of
$\Delta^*$.

\begin{lemma} \label{lem5} Suppose the  adversary is $\rho$-bounded
  for $\rho<1/3$. Then the protocol is live with respect to a liveness
  parameter $T_l$ that is $O(k^* \Delta)$. 
\end{lemma} 

\begin{proof}
The proof shows (essentially) that each honest leader after GST
confirms a new block of transactions, but is complicated slightly by
transitions at the end of an epoch. 
Recall that, by the assumptions of
the quasi-permissionless setting, honest validators are always active.

Suppose the adversary is $\rho$-bounded for $\rho<1/3$.  As noted
above,
we can refer unambiguously to the value $\mathtt{lead}(e,v)$ for each
$e,v$.  Suppose that the environment sends a transaction $\mathtt{tr}$
to an honest player at $t$ and let
$t_0=\text{max} \{ t, \text{GST} \} +\Delta$. Let $e$ be the greatest
epoch that any honest player is in at $t_0$, and suppose that all
active honest players are in view $v_0$ at $t_0$. Let $v_1>v_0$ be the
least such that $\mathtt{lead}(e,v_1)$ is honest (noting that
$v_1\leq v_0+k^*$). There are now various possibilities to consider. 
If no honest player enters epoch $e+1$ by timeslot $4v_1\Delta +3\Delta$,
then let $Q^+$ be the lock amongst honest players at time
$4v_1\Delta-\Delta$ such that $\mathtt{v}(Q^+)$ is greatest.  The
leader for view $v_1$ will receive $Q^+$ and $\mathtt{tr}$ by time
$4v_1\Delta$.  Assuming $\mathtt{tr}$ has not already been confirmed,
this leader will therefore propose a block containing $\mathtt{tr}$
that is admissible for all honest players (when they judge
admissibility at timeslot $4v_1\Delta +\Delta$).  All honest players
will then produce stage 1, 2 and 3 votes for that block $b$, and $b$
 will be confirmed by the end of view $v_1$.
 
  If $\mathtt{h}(b)\leq (e+1)x$ then $\mathtt{tr}$ will also be confirmed by the end of view $v_1$, and so the argument is complete in that case. If either $\mathtt{h}(b)> (e+1)x$ or  some honest player enters epoch $e+1$ before
timeslot $4v_1\Delta +3\Delta$,  all honest players will
be in epoch $e+1$ by timeslot $4(v_1+1)\Delta$.
 In this case, let $v_2>v_1+1$ be
the least such that $\mathtt{lead}(e+1,v_2)$ is honest (noting that
$v_2\leq v_1+1+k^*$). Then all honest players will be in epoch $e+1$
for the entirety of view $v_2$.  Let $Q^+$ be the lock amongst honest
players at time $4v_2\Delta-\Delta$ such that $\mathtt{v}(Q^+)$ is
greatest.  The leader for view $v_2$ will receive $Q^+$ and
$\mathtt{tr}$ by time $4v_2\Delta$.  Assuming $\mathtt{tr}$ has not
already been confirmed, this leader will therefore propose a block
containing $\mathtt{tr}$ that is admissible for all honest players at
timeslot $4v_2\Delta +\Delta$.  All honest players will then produce
stage 1, 2 and 3 votes for that block, and the block (and
$\mathtt{tr}$) will be confirmed by the end of view $v_2$.
%
\end{proof}

\subsection{Proving that the recovery procedure carries out slashing}  \label{verrp}

\begin{lemma} \label{lem6} 
Suppose the adversary is $\rho^*$-bounded for $\rho^*<2/3$ and that
all message delays prior to GST are at most $\Delta^*$. If a
consistency violation occurs in some least epoch $e$,  then the
recovery procedure produces a unique updated genesis block $b_g'$ such that:
\begin{itemize} 

\item $\mathtt{Tr}(b_g')$ contains a certificate of guilt that
  implicates at least 1/3 of the validating stake for epoch~$e$.

\item $b_g'$ receives output votes from more than half the validating stake for epoch $e$ that is not implicated by certificates of guilt in  $\mathtt{Tr}(b_g')$. 

\end{itemize}
\end{lemma}

\begin{proof} 
The argument resembles that in
Section~\ref{DSreview}. 
We again use the fact, implied by the assumptions
of the quasi-permissionless setting, that honest validators are always
active.

Assume the conditions in the statement
of the lemma, and let $b$ be an epoch $e-1$ ending block that is
$\mathcal{M}^*$-confirmed. Because $e$ is the least epoch in which a
consistency violation occurs, $b$ is uniquely determined. Set
$\mathtt{T}:=\mathtt{Tr}(b)$. Let $id_0,\dots,id_{k-1}$ be the
enumeration of the identifiers $id$ such that
$\mathtt{S}(\mathtt{S}^*,\mathtt{T},id)>0$ as specified in Section
\ref{rp}, and set $\mathtt{Id}:= \{ id_0,\dots,id_{k-1} \}$. Let $P$
be the set of players with identifiers in $\mathtt{Id}$.  For
$i\in \mathbb{N}_{\geq 0}$, define $\mathtt{reclead}(i)$ to be the
identifier $id_j$, where $j= i \text{ mod } k$. Note that
$\mathtt{reclead}(i)=\mathtt{reclead}(M,i)$ however $M$ is locally
defined during the recovery procedure for honest $p\in P$.

For each $i\in \mathbb{N}_{\geq 0}$, let $t_i$ be as defined in
Algorithm 2, so that the $i^{\text{th}}$ instance of Dolev-Strong
begins at time $t_i$. Note that the recovery procedure might not be
triggered by time $t_i$. If $p=\mathtt{reclead}(i)$ is honest, and if
$p$ has not received a certificate of guilt for epoch $e$ by $t_i$,
then $O_{p'}$ will remain empty for every honest $p'\in P$ during the
$i^{\text{th}}$ instance of Dolev-Strong, and no honest player in $P$
will output a vote in this instance. If $p=\mathtt{reclead}(i)$ is
honest, and if $p$ has received a certificate of guilt for epoch $e$
by $t_i$ and has not terminated the recovery procedure by time $t_i$,
then $p$ will disseminate a genesis proposal $(b_g',i)$ which is
$M$-admissible, where $M$ is as locally defined for any honest
$p'\in P$ during the recovery procedure.  $\mathtt{Tr}(b_g')$ will
contain a certificate of guilt $\mathtt{G}$ implicating at least 1/3
of the validating stake for epoch $e$. Every honest player $p'\in P$
will trigger the recovery procedure by time $t_i+\Delta^*$ and will
enumerate this single value $(b_g',i)$ into $O_{p'}$ during this
instance of Dolev-Strong. All honest players in $P$ will then produce
output votes for $b_g'$ in this instance. Since the honest players
control a majority of the validating stake for epoch $e$ that is not
implicated by $\mathtt{G}$ (using that $\rho^* < 2/3$), $b_g'$
receives output votes from more than half the validating stake for
epoch $e$ that is not implicated by certificates of guilt in
$\mathtt{Tr}(b_g')$.
 
If $p=\mathtt{reclead}(i)$ is Byzantine and if $p'\in P$ is honest and
enumerates some value into $O_{p'}$ during the $i^{\text{th}}$
instance of Dolev-Strong, it remains to show that every honest
$p''\in P$ will do the same. This suffices to show that either all
honest players in $P$ will vote for a common updated genesis block and
end the recovery procedure, or none will.
The argument is the same as in Section~\ref{DSreview}.  There are two
cases to consider:
\begin{itemize} 
\item \textbf{Case 1}. Suppose that some honest $p'\in P$ first
  enumerates $y$ into $O_{p'}$ at time $t_i+ j\Delta^*$ with $j<k$. In
  this case, $p'$ receives a message of the form $m=y_{id_1',\dots
    ,id_j'}\in \mathtt{Signed}(M,i,j)$ at  $t_i+ j\Delta^*$. Player
  $p'$  then adds their signature to form a message with $j+1$
  distinct signatures and disseminates this message. This means every
  honest player $p''\in P$  will enumerate $y$ into~$O_{p''}$   by
  time $(j+1) \Delta^*$. 

\item \textbf{Case 2}. Suppose next that some honest $p'$ first enumerates $y$ into $O_{p'}$ at time $t_i+k\Delta^*$. In this case, $p'$ receives a message of the form $m=y_{id_1',\dots ,id_k'}\in \mathtt{Signed}(M,i,k)$ at time $t_i+k\Delta^*$, which has $k$ distinct signatures attached. At least one of those signatures must be from an honest player $p''$, meaning that Case 1 applies w.r.t.\  $p''$.  
\end{itemize} 
\end{proof} 

\vspace{0.2cm}
\noindent \textbf{A further relaxation of the synchronicity assumption.}
As foreshadowed in footnote~\ref{foot:limited},
bounded message delays are required only during a
  limited period around the time of an attack.
Specifically,
suppose there exists a consistency violation in some
least epoch $e$. Let $t_0^*$ be the least timeslot at which some
honest player enters epoch $e+1$ or triggers the recovery
procedure. While the statements of Lemmas \ref{lem2} and \ref{lem6}
assume for the sake of simplicity that message delays prior to GST are
always bounded by $\Delta^*$, the following weaker assumption suffices
for the proofs:
If $p$ disseminates a message at any timeslot $t_1^*$, and if
honest $p'$ is active at
$t_2^*:= \text{max} \{ t_1^*, t_0^*\} +\Delta^*$, then~$p'$ receives
that message at a timeslot $\leq t_2^*$.

\subsection{Completing the proof of Theorem~\ref{posres}:
  Establishing the EAAC property} \label{instant}

To complete the proof of Theorem \ref{posres}, we must specify $t_f$, as well as the investment and valuation functions. 

\vspace{0.2cm} 
\noindent \textbf{Defining $t_f$}. If there is a consistency violation in some
least epoch $e$, let $t_f$ be any timeslot after which all honest
players active during epoch $e$ have terminated the recovery procedure
and outputted an updated genesis block $b_g'$. By Lemma
\ref{lem6}, such a timeslot $t_f$ exists, and $b_g'$ is uniquely
defined and can be determined from the locally defined $M$ values of
active honest players at  $t_f$.

\vspace{0.2cm} 
\noindent \textbf{The canonical PoS investment function}. Let $C$ be
the cost of each unit of stake. To specify $\mathtt{R}(p,t)$, we
define $\mathtt{R}(id,t)$ for each identifier $id$ and then set
$\mathtt{R}(p,t)=\sum_{id\in \mathtt{id}(p)} \mathtt{R}(id,t)$. 
The rough idea is that a player's investment should correspond to the
value of their validating stake, the value of any of
their stake which is no longer validating but remains in escrow,
and the value of any of their stake which was lost to slashing.
Formally, for a timeslot~$t$, 
denote by $e(t)$ the greatest epoch 
such that, for all $e\leq e(t)$, there exists a unique epoch $e$
ending block that is confirmed for some honest player at $t$. 
(If no such epoch exists, define~$b_1(t)$ and $b_2(t)$ as the genesis
block.)
Let $b_1(t)$ be the unique epoch $e(t)$ ending block
that is confirmed for some honest player at $t$, and
$b_2(t)$ the unique epoch $e(t)-1$ ending block that is confirmed for
some honest player at $t$. (Or if~$e(t)=0$, define~$b_2(t)$ as the
genesis block.)
First, suppose that no honest player has triggered the recovery
procedure by $t$. Then, if either:
\begin{itemize}
\item [(i)]
the identifier is currently a validator (i.e., $\mathtt{S}(\mathtt{S}^*,\mathtt{Tr}(b_1(t)),id)>0$);
\item [(ii)]
the identifier no longer controls a validator but their stake
  has not yet been removed from escrow (i.e., $\mathtt{S}(\mathtt{S}^*,\mathtt{Tr}(b_2(t)),id)>0$); or
\item [(iii)]
the identifier has been slashed (i.e., $\mathtt{Tr}(b_1(t))$
  includes a certificate of guilt   implicating~$id$),
\end{itemize}
we define $\mathtt{R}(id,t)=N/x^*$; otherwise, $\mathtt{R}(id,t)=0$.

Now suppose that some honest player has triggered the recovery
procedure by $t$, due to a consistency violation in
epoch~$e(t)+1$. Here, because the validating stake from epoch $e(t)-1$
need not still be in escrow, we consider only~$b_1(t)$. That is,
we define $\mathtt{R}(id,t)=N/x^*$ if~(i) or~(iii) holds, and
$\mathtt{R}(id,t)=0$ otherwise.

This investment function is $\Gamma$-liquid in the sense of
Section~\ref{ss:eaac}, provided $\Gamma$ is chosen to be larger than
the maximum time needed to complete an epoch after GST (which is
proportional to $\Delta \cdot x > \Delta^*$).

\vspace{0.2cm} 
\noindent \textbf{The canonical valuation function}. If there is no
consistency violation at timeslots $\leq t$, then, for any set of
players $P$, we define the valuation function as the value (at the
market price~$C$) of the investments made (and not already recouped)
by~$P$:
\[ v(P,t,\{ \mathtt{R}(p,t) \}_{p \in P},C, \mathcal{M})):= \sum_{p\in
    P} \mathtt{R}(p,t) \cdot C, \]
where $\mathcal{M}$ specifies the sets of messages received by each
player by timeslot~$t$ in the execution.

If the event of a consistency violation, the valuation function should
discount stake that has been slashed.  Formally, consider a
consistency violation and let $t_f$ and $b_g'$ be as specified above.
Let $\mathtt{Id}$ be the set of identifiers that are not implicated by
any certificates of guilt in $\mathtt{Tr}(b_g')$.  For any set of
players $P$ and $t \ge t_f$, we set:
\[ v(P,t,\{ \mathtt{R}(p,t) \}_{p \in P},C, \mathcal{M})):= \sum_{p\in
    P} \sum_{id\in \mathtt{id}(p) \cap \mathtt{Id}} \mathtt{R}(id,t)
  \cdot C,  \]
where $\mathcal{M}$ specifies the sets of messages received by each
player by timeslot~$t$ in the execution. 
This valuation function is consistency-respecting in the sense of
Section~\ref{ss:eaac}. (It is generally undefined following a
consistency violation, up until the timeslot~$t_f$ at which the recovery
procedure completes.)


Because Lemma \ref{lem1} establishes that certificates of
guilt can implicate only Byzantine identifiers, condition (a) in
Definition \ref{d:eaac} is satisfied.  By Lemma \ref{lem6}, in the
case of a consistency violation in some least epoch $e$, $\mathtt{Tr}(b_g')$
contains certificates of guilt for at least
1/3 of the validating stake for epoch $e$, so condition (b) in
Definition \ref{d:eaac} is also satisfied and the protocol is
$(\alpha,\rho,\rho^*)$-EAAC for 
$\alpha = \max\{0, (\rho^*-\tfrac{1}{3})/\rho^* \}$.
$\qed$

\vspace{0.2cm} 
\noindent \textbf{Weakening liveness requirements to achieve the EAAC property for larger adversaries}. Thus far, we have supposed that the protocol is required to be live so long as the adversary is $\rho$-bounded for $\rho<1/3$.  If the protocol  is only required to be live for adversaries that are $\rho_l$-bounded (for some $\rho_l<1/3$), then a simple modification suffices to ensure the protocol is EAAC with respect to the canonical investment and valuation functions  when the adversary is $\rho$-bounded for  $\rho<1-\rho_l$. To describe the appropriate modification, consider condition (iv) from the definition of a quorum certificate in Section \ref{formal}, and suppose we require instead that  $\sum_{V\in Q} \mathtt{c}(V)\geq  (1-\rho_l) \mathtt{N}(b)$. For the proof of Lemma \ref{lem2}, we  required that every quorum certificate must
contain at least one vote from an honest player.  With our modified notion of a quorum certificate, these conditions will now hold so long as the adversary is $\rho$-bounded for  $\rho<1-\rho_l$. 
Since any certificate of guilt
now implicates at least a fraction $1-2\rho_l$ of the validating stake for epoch $e$, Lemma \ref{lem3} will now hold under the weaker condition that the adversary is $\rho$-bounded for $\rho<1-2\rho_l$. 
The proof of Lemma \ref{lem5} goes through unchanged, so long as the adversary is $\rho_l$-bounded. The proof of Section \ref{verrp} also goes through unchanged.

\vspace{0.2cm} 
\noindent \textbf{How is the impossibility result in Theorem
  \ref{neg2} avoided?}  As discussed at the end of
Section~\ref{neg2proof}, the proof of Theorem~\ref{neg2} shows that
non-trivial EAAC protocols are possible only if the protocol's
liveness parameter or the investment function's liquidity parameter
scales with the worst-case message delay.  In Theorem~\ref{posres},
crucially, we do assume a finite bound~$\Delta^*$ on the worst-case
message delay.  The liveness parameter~$T_l$ of the protocol scales
only with~$\Delta$ and the maximum number~$x^*$ of validators, and
thus may be arbitrarily smaller than the worst-case message
delay~$\Delta^*$.  The canonical PoS investment function, however,
is~$\Gamma$-liquid only for values~$\Gamma$ that exceed the minimum
length of an epoch, which is, by our choice of the number~$x$ of
blocks per epoch, at least $4\Delta x > 4\Delta^*$ timeslots.  This
delay, combined with our use of three stages of voting (and using the
fact that $\rho^*<2/3$), allows us to ensure that any consistency
violation is seen by honest players before those responsible for it
can cash out of their positions (with honest
players learning of a consistency violation in epoch $e$ within
$2\Delta^*$ timeslots of any honest player completing the epoch).

\subsection*{Acknowledgments}

The research of the third author at Columbia University is supported
in part by NSF awards
CCF-2006737 and CNS-2212745, and research awards from the Briger
Family Digital Finance Lab and the Center
for Digital Finance and Technologies.

\bibliographystyle{plainurl}

\appendix

\section{Modeling PoW in Bitcoin} \label{aa:bitcoin} 

In this section, we show how to model proof-of-work, as used in the Bitcoin
protocol, using the permitter formalism.  The basic idea is that, at
each timeslot, a player should be able to request permission to
publish a certain block, with the permitter response depending on the
player's resource balance (i.e., hashrate).  To integrate properly
with the ``difficulty adjustment'' algorithm used by the Bitcoin
protocol to regulate the average rate of block production (even as the
total hashrate devoted to the protocol fluctuates), we define the
permitter oracle so that its responses can have varying ``quality.''
Each PoW oracle response will be a 256-bit string~$\tau$, and we can
regard the quality of $\tau$ as the number of 0s that it begins with.

\vspace{0.2cm} 
\noindent \textbf{An example}. As an example, suppose that
$\mathtt{R}^O(p,t)=5$. Then $p$ might send a single query
$(5,\sigma)$ at timeslot $t$, and this should be thought of as a
request for a PoW for $\sigma$ (which could represent a block
proposal, for example).  If $p$ does not get the response that it's
looking for, it might,
at a later timeslot $t'$ with $\mathtt{R}^O(p,t')=5$, 
send another query $(5,\sigma)$ to the PoW oracle.

\vspace{0.2cm} 
\noindent \textbf{Formal details}.  We consider a single-use
permitter. All queries to the permitter must be of the form
$(b,\sigma)$, where $\sigma$ is any finite string.  As already
stipulated in Section \ref{mec}, the fact that $O$ is a single-use permitter means
that $p$ can send the queries
$(b_1,\sigma_1),\dots (b_k,\sigma_k)$ at the same timeslot $t$ only if
$\sum_{i=1}^k b_i\leq \mathtt{R}^O(p,t)$. If $p$ submits the query $(b,\sigma)$ to the permitter at $t$, then the permitter independently samples
$b$-many 256-bit strings uniformly at random and lets $\tau$ be the
lexicographically smallest of these sampled strings. Player $p$ then receives the response  $r:=(\sigma,\tau)$ signed by the oracle.

\end{document}